\documentclass[reqno,11pt]{amsart}
\usepackage{amsmath, latexsym, amsfonts, amssymb, amsthm, amscd, array, caption}
\usepackage[utf8]{inputenc}
\usepackage{graphics,epsf,psfrag,epsfig}
\usepackage[colorlinks=true, pdfstartview=FitV, linkcolor=blue, citecolor=blue, urlcolor=blue,pagebackref=false]{hyperref}

\usepackage[showlabels,sections,floats,textmath,displaymath]{}

\numberwithin{equation}{section}
\newtheorem{theorem}{Theorem}
\newtheorem{lemma}{Lemma}[section]
\newtheorem{proposition}[theorem]{Proposition}

\newtheorem{rem}{Remark}[section]
\newtheorem{claim}{Claim}[section]

\newtheorem{hyp}{Assumption}

\renewenvironment{proof}[1][Proof]{\begin{trivlist}
\item[\hskip \labelsep {\bfseries #1}]}{\qed\end{trivlist}}

\DeclareMathOperator{\p}{\mathbb{P}}

\newcommand{\ind}{\mathbf{1}}
\renewcommand{\ge}{\geq}
\renewcommand{\le}{\leq}
\newcommand{\R}{\mathbb{R}}

\newcommand{\N}{\mathbb{N}}
\renewcommand{\tilde}{\widetilde}
\renewcommand{\hat}{\widehat}

\DeclareMathSymbol{\leqslant}{\mathalpha}{AMSa}{"36} 
\DeclareMathSymbol{\geqslant}{\mathalpha}{AMSa}{"3E} 
\DeclareMathSymbol{\eset}{\mathalpha}{AMSb}{"3F}     
\renewcommand{\leq}{\;\leqslant\;}                   
\renewcommand{\geq}{\;\geqslant\;}                   
\newcommand{\dd}{\,\text{\rm d}}             

\newcommand{\sumtwo}[2]{\sum_{\substack{#1 \\ #2}}} 
\newcommand{\floor}[1]{\lfloor #1 \rfloor}

\newcommand{\cA}{{\ensuremath{\mathcal A}} }

\newcommand{\cI}{{\ensuremath{\mathcal I}} }

\newcommand{\bP}{{\ensuremath{\mathbf P}} }
\newcommand{\bE}{{\ensuremath{\mathbf E}} }

\newcommand{\bbE}{{\ensuremath{\mathbb E}} }

\newcommand{\bbN}{{\ensuremath{\mathbb N}} }

\newcommand{\bbP}{{\ensuremath{\mathbb P}} }

\newcommand{\bbR}{{\ensuremath{\mathbb R}} }

\newcommand{\ga}{\alpha}
\newcommand{\gb}{\beta}
\newcommand{\gga}{\gamma}            
\newcommand{\gd}{\delta}
\newcommand{\gep}{\varepsilon}       

\newcommand{\gk}{\kappa}
\newcommand{\go}{\omega}

\newcommand{\gl}{\lambda}

\newcommand{\gU}{\Upsilon}

\makeatletter
\def\captionfont@{\footnotesize}
\def\captionheadfont@{\scshape}

\long\def\@makecaption#1#2{%
  \vspace{2mm}
  \setbox\@tempboxa\vbox{\color@setgroup
    \advance\hsize-6pc\noindent
    \captionfont@\captionheadfont@#1\@xp\@ifnotempty\@xp
        {\@cdr#2\@nil}{.\captionfont@\upshape\enspace#2}%
    \unskip\kern-6pc\par
    \global\setbox\@ne\lastbox\color@endgroup}%
  \ifhbox\@ne 
    \setbox\@ne\hbox{\unhbox\@ne\unskip\unskip\unpenalty\unkern}%
  \fi
  \ifdim\wd\@tempboxa=\z@ 
    \setbox\@ne\hbox to\columnwidth{\hss\kern-6pc\box\@ne\hss}%
  \else 
    \setbox\@ne\vbox{\unvbox\@tempboxa\parskip\z@skip
        \noindent\unhbox\@ne\advance\hsize-6pc\par}%
\fi
  \ifnum\@tempcnta<64 
    \addvspace\abovecaptionskip
    \moveright 3pc\box\@ne
  \else 
    \moveright 3pc\box\@ne
    \nobreak
    \vskip\belowcaptionskip
  \fi
\relax
}
\makeatother
\def\writefig#1 #2 #3 {\rlap{\kern #1 truecm
\raise #2 truecm \hbox{#3}}}

\newcommand{\q}{\mathrm{que}}
\newcommand{\tf}{\mathtt{F}}
\renewcommand{\P}{{\ensuremath{\mathbf P}} }
\newcommand{\E}{{\ensuremath{\mathbf E}} }

\newcommand{\hca}{h_{c}^{\a}}
\renewcommand{\a}{ \mathrm{a}}

\newcommand{\K}{\mathrm{K}}
\newcommand{\Znc}{Z_{n,h_c^\a}^{\a}}
\newcommand{\Enc}{\ensuremath{\mathbf E}_{n,h_c}^{\a}}

\newcommand{\Zna}{Z_{n,h}^{\a}}

\newcommand{\hP}{\hat\P}

\newcommand{\W}{\mathsf{W}}

\author[Q. Berger]{Quentin  Berger
\\
\textit{\tiny University of Southern California}}

\title[Pinning model in long-range correlated Gaussian environment]{Comments on the influence of disorder
for pinning model in correlated Gaussian environment}

\begin{document}

\begin{abstract}
  We study the random pinning model, in the case of
  a Gaussian environment
  presenting power-law decaying correlations, of exponent decay $a>0$.
  A similar study was done in a hierachical version of the
  model \cite{BThier},  and we extend here the results to
  the non-hierarchical (and more natural) case.
  We comment on  the annealed
  (\textit{i.e.}\ averaged over disorder) model,
  which is far from being trivial, and we discuss the influence of disorder
  on the critical properties of the system.
  We show that the annealed critical
  exponent $\nu^{\rm a}$ is the same as
  the homogeneous one $\nu^{\rm pur}$,
  provided that correlations are decaying
  fast enough ($a>2$).
  If correlations are summable ($a>1$),
  we also show that the disordered phase
  transition is at least of order $2$,
  showing disorder relevance if $\nu^{\rm pur}<2$.
  If correlations are not summable ($a<1$),
  we show that the phase transition disappears.
\end{abstract}

\maketitle

\section{Introduction}
\label{sec:correlintro}

The question of the influence of inhomogeneities on the critical properties
of a physical system has been studied
in the physics literature for a great variety of models. 
In the case where the disorder is IID, the question of
relevance/ir\-re\-le\-vance of disorder is predicted
by the so-called {\sl Harris criterion} \cite{Harris}: disorder is irrelevant if  $\nu^{\rm pur}>2$,
where $\nu^{\rm pur}$ is the correlation length critical exponent of the homogeneous model.
Following the reasoning of
Weinrib and Halperin \cite{WeinHalp83}
one realizes that,
introducing correlations with power-law decay
$r^{- a }$ (where $ a >0$, and $r$ the distance between the points),
disorder should be relevant if $\nu^{\rm pur}<2/ \min(a,1) $,
and irrelevant if $\nu^{\rm pur}>2/ \min(a,1) $.
Therefore, the Harris prediction for disorder relevance/ir\-re\-le\-vance
should be modified only if $ a <1$.

In the mathematical literature,
the question of disorder (ir)relevance has been
very active during the past few years,
in the framework of polymer pinning models \cite{denHoll,GBbook,SFLN}.
The Harris criterion
has in particular been proved thanks to a series of articles.
We investigate here
the polymer pinning model in random correlated environment
of Gaussian type, with correlation decay exponent $a>0$.
Several results where obtained in \cite{BThier}, for the hierachical pinning model, and we prove here a variety of corresponding results on the disordered and annealed non-hierarchical system.
In particular, we confirm part of the Weinrib-Halperin prediction for $ a >1$. We also show that the case $ a <1$ is somehow special, and that the behavior of the system does not fit the prediction in that case.

\subsection{The disordered pinning model}

Consider $\tau:=(\tau_n)_{n\geq0}$ a recurrent renewal process, with
law denoted by $\bP$: $\tau_0=0$, and the $(\tau_i-\tau_{i-1})_{i\geq1}$
are IID, $\bbN$-valued.
The set $\tau=\{\tau_0,\tau_1,\ldots\}$ (making a slight abuse of notation)
can be thought as the set of contact points between a polymer and a defect line.
We assume that
the inter-arrival distribution
 $\K(\cdot)$ verifies
\begin{equation}
\K(n):=\bP(\tau_1=n)= \frac{\varphi(n)}{n^{1+\ga}},
\label{assumpK}
\end{equation}
for some $\ga\geq 0$, and slowly varying function $\varphi(\cdot)$ (see \cite{Bingham}).
The fact that the renewal is recurrent simply means that
$\K(\infty)=\bP(\tau_1=+\infty)=0$.
We also assume for simplicity that $\K(n)>0$ for all $n\in \N$.

Given a sequence $\go=(\go_n)_{n\in \N}$ of real numbers (the environment),
and parameters $h\in \bbR$
and $\gb\ge 0$,
we define the {\sl polymer} measure $\bP_{N,h}^{\go,\gb}$, $N\in \N$, as follows
\begin{equation}
 \frac{\dd \bP_{N,h}^{\go,\gb}}{\dd \bP} (\tau) := 
\frac{1}{Z_{N,h}^{\go,\gb}}   \exp\left( \sum_{n=1}^N  (h+\gb\go_n)\gd_n \right) \gd_N,
\end{equation}
where we noted $\gd_n:=\ind_{\{n\in\tau\}}$, and where
$ Z_{N,h}^{\go,\gb}:= \bE\left[ \exp\left( \sum_{n=1}^N  (h+\gb\go_n) \gd_n\right) \gd_N \right]$
is the \emph{partition function} of the system.

\smallskip
In what follows, we take $\go$ a random ergodic sequence, with law denoted by $\bbP$.
We also assume that $\go_0$ is integrable.
\begin{proposition} [see \cite{GBbook}, Theorem 4.6]
\label{prop:Fque}
The limit
\begin{equation}
 \tf(\gb,h):= \lim_{N\to\infty} \frac{1}{N} \log Z_{N,h}^{\go,\gb}
= \sup_{N\in\N} \frac{1}{N} \bbE \log Z_{N,h}^{\go,\gb},
\end{equation}
exists and is constant $\bbP$ a.s.\ It is called the \textsl{quenched} free energy.
There exists a quenched critical point
$h_c^{\q}(\gb)\in\bbR$, such that $\tf(\gb,h)>0$ if and only if $h>h_c^{\q}(\gb)$. 
\end{proposition}

We stress that the free energy carries some physical information
on the thermodynamic limit of the system.
Indeed, one has that at every point where $\tf$ has a derivative, one has
\begin{equation}
\label{contacts}
\lim_{N\to\infty} \frac1N \bE_{N,h}^{\go,\gb}\left[ \sum_{n=1}^N \gd_n  \right]
 = \frac{\partial}{\partial h} \tf(\gb,h).
\end{equation}
Therefore, thanks to the convexity of $h\mapsto\tf(\gb,h)$,
one concludes that if $\tf(\gb,h)>0$
there is a positive density of contacts under the polymer measure,
in the limit $N$ goes to infinity.
Then the critical point $h_c^{\q}(\gb)$
marks the transition between the delocalized phase (for $h<h_c^{\q}$, $\tf(\gb,h)=0$)
and the localized phase (for $h>h_c^{\q}$, $\tf(\gb,h)>0$).

One also defines the annealed partition function,
$ Z_{N,h,\gb}^{\a}:=\bbE[Z_{N,h}^{\go,\gb}]$, used to be confronted to
the disordered system.
Then the {\sl annealed} free energy is defined as
  $\tf^{\a}(\gb,h):=\lim_{N\to\infty} \frac{1}{N}\log  \bbE Z_{N,h}^{\go,\gb}$,
and one has an {\sl annealed} critical point $h_c^{\a}(\gb)$ that separates
phases where $\tf^{\a}(\gb,h)=0$ and where $\tf^{\a}(\gb,h)>0$.
A simple use of Jensen's inequality yields that
$\tf(\gb,h)\leq\tf^{\a}(\gb,h)$, so that $h_c^{\a}(\gb)\geq h_c^{\q}(\gb)$.

\subsubsection{The homogeneous model}

The homogeneous pinning model is the pinning model
with no disorder, \textit{i.e.}\ with $\gb=0$. The partition function is
$Z_{N,h}:=\bE\left[ e^{h\sum_{n=1}^N\gd_n}\gd_N\right]$.
This model is actually fully solvable.

\begin{proposition}[\cite{GBbook}, Theorem 2.1]
 The \textsl{pure free energy},
$ \tf(h):= \tf(0,h)$, exhibits a phase transition at the
critical point $h_c=0$ (recall we have a recurrent renewal $\tau$).
One has the following asymptotic of $\tf(h)$ around $h=0_+$: for every $\alpha\geq 0$ and $\varphi(\cdot)$, there exists some slowly varying function $\hat\varphi(\cdot)$ such that
\begin{equation}
 \tf(h) \stackrel{h\searrow0}{\sim} \hat\varphi(h) h^{1\vee 1/\ga},
\end{equation}
where $f\sim g$ means that the ratio $f/g$ converges to $1$, and $a\vee b$
stands for the maximum between $a$ and $b$.
\label{comphomo}
\end{proposition}
The pure critical exponent is therefore $\nu^{\rm pur}:=1\vee1/\ga$,
and it encodes the critical behavior of the homogeneous model.

\subsection{The case of an IID environment}

First, note that in the IID case, the annealed partition function is
$\bbE\left[ e^{\sum (\gb\go_n+h) \gd_n} \right] = \bbE\left[ e^{\sum (\gl(\gb)+h) \gd_n} \right]$
with $\gl(\gb):= \log \bbE\left[  e^{\gb\go_1} \right]$:
the annealed system
is the homogeneous pinning model with parameter $h+\gl(\gb)$,
and is therefore understood. In particular, the annealed
critical point is $h_c^{\a}(\gb)=-\gl(\gb)$.

For the pinning model in IID environment,
the Harris criterion for disorder relevance/ir\-re\-le\-vance
is mathematically settled, both in terms of
critical points and in terms of critical exponents.
A recent series of papers indeed proved that

\textbullet\ if $\ga<1/2$, then disorder is irrelevant: if $\gb>0$ is small enough, one has that
$h_c^{\q}(\gb)=h_c^{\a}(\gb)$,
and the quenched critical behavior is the same as the homogeneous one;

\textbullet\ if $\ga\geq 1/2$, then disorder is relevant: for any $\gb>0$ one has
$h_c^{\q}(\gb)>h_c^{\a}(\gb)$,
and the order of the disordered phase transition is at least $2$ (thus strictly larger than
$\nu^{\rm pur}$ if $\ga>1/2$).

We refer to \cite{A06,AZ08,CdH10,GLT09,GT05,GT_ap,L10,T_aap,T08}
for specific details, and \cite{SFLN} for a review of the techniques used.

\subsection{The long-range correlated Gaussian environment}

Up to recently, the pinning model defined above was studied only in an IID  environment,
or in the case of a Gaussian environment with
finite-range correlations \cite{Poisat1, Poisat2}. In this latter case, it is
shown that the features of the system are the same as with an IID  environment,
in particular concerning the disorder relevance picture.
In \cite{Bstrong,BLpin},
the authors study the drastic effects of the presence of large and frequent
attractive regions
on the phase transition:
important disorder fluctuations
lead to a regime where disorder always modifies the critical properties, whatever $\nu^{\rm pur}$ is.
In \cite{BThier,Poisat2},
the authors focus on long-range correlated Gaussian environment, as we now do.

Let $\go=(\go_n)_{n\in\N}$ be a Gaussian stationary process (with law $\bbP$), with zero mean
and unitary variance,
and with correlation function $(\rho_n)_{n\geq0}$.
We denote the covariance matrix $\gU=(\gU_{ij})_{i,j\in\bbN}$
(with the notation $\gU_{ij}:=\bbE[\go_i\go_j]=\rho_{|j-i|}$),
which is symmetric definite positive (so that $\go$ is well-defined).
We also assume that $\lim_{n\to\infty}|\rho_n|=0$,
so that the sequence $\go$ is ergodic (see \cite[Ch.14 \S2, Th.2]{Cornfeld}).

\smallskip
The Weinrib-Halperin prediction
suggests to consider a power-law decaying correlation function,
$\rho_n \sim n^{- a }$, and we therefore make the following assumption.

\begin{hyp}
\label{hyphyp}
There exist some $ a >0$ and a constant $c_0>0$ such that
\begin{equation}
 \rho_{k}\stackrel{k\to\infty}{\sim} c_0  k^{- a }.
\end{equation}
We refer to \emph{summable correlations} when $a>1$, and to \emph{non-summable correlations} when $a\leq 1$.

We stress here that $\rho_k=(1+k)^{-a}$ for all $k\geq 0$ is a valid choice for a correlation function, since it is convex, cf. \cite{Polya}.
\end{hyp}

Note that most of our results are actually valid under more general assumptions, but we focus on this Assumption, which is very natural and make our statements clearer.



\subsection{Comparison with the hierarchical framework}
\label{sec:comparhier}

In \cite{BThier}, the authors focus
on the hierarchical version of the pinning model, and we believe that
all the results they obtain should have an analogue
in the non-hierarchical framework.
In \cite{BThier}, the correlations respect the hierarchical structure:
${\rm Cov}(\go_i,\go_j)= \gk^{d(i,j)}$, where $d(i,j)$ is the
hierarchical distance between $i$ and $j$.
It corresponds to a power law decay $|i-j|^{- a }$ in the non-hierarchical model,
with $a  := \log( 1/\gk)/\log 2$ (we keep this notation for this section).
We therefore compare our model with the hierarchical one,
and give more 
predictions on the behavior of the system, and on the
influence of correlations on the disorder relevance picture: see Figure \ref{fig:zones},
in comparison with \cite[Fig. 1]{BThier}.

In the hierarchical framework,
different behaviors have been identified:

\textbullet\ {If $a>1$, $a \nu^{\rm pur} > 2$.} Then one controls the annealed model
close to the annealed critical point (see \cite[Prop 3.2]{BThier}): in particular the annealed critical behavior is the same as the homogeneous one, $\nu^{\a}=\nu^{\rm pur}$.
In this region, the Harris criterion is not modified:
\begin{itemize}
 \item[-] If $\nu^{\rm pur}>2$, then disorder is \textsl{irrelevant}: there exists some $\gb_0>0$ such
that $h_c(\gb)=h_c^{\a}(\gb)$ for any $0<\gb\leq \gb_0$.
Moreover, for every $\eta>0$ and choosing $u>0$ sufficiently small,
$\tf(\beta,\hca(\beta)+u)\ge (1-\eta)\tf^{\rm a}(\beta,\hca(\beta)+u)$, so that $\nu^{\q}=\nu^{\a}=\nu^{\rm pur}$.

 \item[-] If $\nu^{\rm pur}\leq2$, then disorder is \textsl{relevant}:
the quenched and annealed critical points differ for every $\gb>0$.
Moreover, the disordered phase transition is at least of order $2$,
so that disorder is relevant in terms of critical exponents if $\nu^{\rm pur}<2$.
\end{itemize}

\textbullet\ {If $a>1$, $a \nu^{\rm pur} < 2$.} Then it is shown that the annealed critical properties are
different than that of the homogeneous model (see \cite[Theorem 3.6]{BThier}).
However, the disordered phase transition is still of order at least $2$,
showing disorder relevance (since $\nu^{\rm pur}<2/a<2$).
\smallskip

\textbullet\ {If $a<1$.} The phase transition does not survive:
the free energy is positive for all values of $h\in\bbR$ as soon $\gb>0$, so that $h_c(\gb)=-\infty$.
It is therefore more problematic to deal with the question of the influence of disorder on the critical properties of the system.

\smallskip

Let us remark that hierarchical model have always been a fruitful tool in the study of disordered systems. In particular, Dyson \cite{Dyson}, from his study of the \emph{hierarchical} ferromagnetic Ising model, combined with the Griffith correlation inequalities, deduced a criteria for the existence of a phase transition for the (non-hierachical) one dimensional ferromagnetic Ising model with couplings $J_{i-j}\sim|i-j|^{-a}$.
We stress that there are no such correlation inequalities for the pinning model, and results cannot be derived directly from the hierarchical model, even though we expect the behavior of the two models to be similar. Therefore, to prove results in the non-hierachical case, we need to adapt the techniques of \cite{BThier}. Many difficulties arise in this process, in particular because the hierarchical correlation structure is much simpler to study than in the non-hierarchical case.

\medskip

Let us highlight how the remaining of the paper is organized.
In Section \ref{sec:respincorrel} we present our main results
on the model and comment them, as well for the annealed system (Theorem \ref{thm:expo_ann})
as for the disordered one (Theorems \ref{thm:smooth}-\ref{thm:xi<1}).
In Section \ref{sec:anncorrel} we collect some crucial observations on the annealed model
in the correlated case, and prove Theorem \ref{thm:expo_ann}.
In Section \ref{sec:proofcorrel} we prove the results on the disordered system.
Gaussian estimates are given in Appendix.

\section{Main results}
\label{sec:respincorrel}

\subsection{The annealed model}

We first focus on the study of the {\sl annealed} model,
which is often the first step towards the understanding of the disordered model.
The annealed partition function is given, thanks to a Gaussian computation, by
\begin{equation}
 \begin{split}
  Z_{N,h}^{\a,\gU}&:= \bbE[ Z_{N,h}^{\go,\gb}] =  \E\left[ e^{H_{N,h}^{\a,\gU}}\, \gd_N\right],\\ 
 \text{with }\quad H_{N,h}^{\a,\gU}&:=(\gb^2/2+h)\sum_{n=1}^N \gd_n +
            \gb^2 \sum_{n=1}^N  \gd_n \sum_{k=1}^{N-n}\rho_k  \gd_{n+k}  .
 \end{split}
\label{defZann}
\end{equation}
We keep the superscript $\gU$ in $Z_{N,h}^{\a,\gU}$,
to recall the correlation structure,
but we drop it if there is no ambiguity.

One remarks that \eqref{defZann} is far from being the partition function
of the standard homogeneous pinning model.
It explains the difficulty of studying the pinning model in correlated random environment:
even annealing techniques,
that give simple and non-trivial bounds in the case of an IID  environment
(where the annealed model is the standard homogeneous one),
are not easy to apply.

The annealed model is actually interesting in itself, since it gives an example of a non-disordered pinning
model in which the rewards correlate according to the position of the renewal points.
One can also consider the annealed model
as a ``standard'' homogeneous pinning model (in the sense that 
a reward $h$ is given to each contact point), but with an underlying correlated
renewal process, that is with non-IID  inter-arrivals.
This model, and in particular its phase transition, is in particular the focus of \cite{Poisat12}.

\begin{proposition}
\label{prop:anncorrel}
If $a>1$, then the limit
\begin{equation}
\label{deffann}
 \tf^{a,\gU}(\gb,h):= \lim_{N\to\infty} \frac{1}{N}\log Z_{N,h}^{\a,\gU}
\end{equation}
 exists, is non-negative and finite. There exists a critical point
$h_c^{\a,\gU}(\gb)\in\bbR$,
such that $\tf^{\a,\gU}(\gb,h)>0$ if and only if $h>h_c^{\a,\gU}(\gb)$. 
\end{proposition}
This result relies on Hammersley's generalized super-additive Theorem \cite[Theorem 2]{Hammer}, and appears in \cite{Poisat12}. We do not prove it here.
One actually only needs the absolute summability of correlations ($\sum_{n\in\bbN} |\rho_n|<+\infty$) to get this proposition. We are unable to tell if this condition is necessary, or if conditionally summable correlations (that is with $\sum_{n\in\bbN} \rho_n<+\infty$ but $\sum_{n\in\bbN} |\rho_n|=+\infty$)
would be sufficient to provide the existence of the annealed free energy.


As far as the annealed critical point is concerned,
an analytic expression is given for $h_c^{\a,\gU}(\gb)$ in \cite{Poisat12}:
it is the maximal eigenvalue of a Ruelle-Perron-Frobenius operator related
to the model
(see \cite[Cor. 4.1]{Poisat12}). However, it is in general not possible to compute its value.
One however gets large-temperature
asymptotic ($\gb\searrow 0$), \cite[Theorem 2.3]{Poisat12}
\begin{equation}
\label{asymppointann}
h_c^{\a,\gU}(\gb) \stackrel{\gb\searrow 0}{\sim } -\frac{\gb^2}{2} 
\left( 1+ 2 \sum_{n\geq 1} \rho_n \P(n\in\tau)\right).
\end{equation}

The following theorem states that if
$a>2$, then the annealed free energy has the same critical exponent
as the pure free energy. This is analogous to \cite[Theorem 3.1]{BThier} in the hierarchical framework.
\begin{theorem}
\label{thm:expo_ann}
Under Assumption \ref{hyphyp}, we suppose that
$a>2$. Then there exist some $\gb_0>0$
and a constant $c_1>0$, such that for any fixed $\gb\leq \gb_0$ one has
\begin{equation}
 \tf\big( c_1^{-1}  u \big)\leq \tf^{\a,\gU}(\gb,h_c^{\a,\gU}(\gb)+u) \leq \tf\big(c_1 u\big),
\end{equation}
for all $u\leq c_1^{-1}$.
\end{theorem}

A analogous result has also independently been proved in \cite{Poisat12} (see Theorem 2.1), using a Ruelle-Perron-Frobenius operator approach to the study of the annealed partition function.
Our proof, however, is (almost completely) self-contained,
and uses basic arguments.

The assumption $a>2$ (that could be weakened to only having $\sum k|\rho_k|<+\infty$) enables us to get some quasi-renewal property for the partition function, see \eqref{quasirenew2}-\eqref{quasirenew}.
We prove Theorem \ref{thm:expo_ann} in Section \ref{sec:behavann}, using this quasi-renewal property.
It is therefore difficult to go beyond the condition $a>2$, since without it, the correlations spread easily from one block to another (see\eqref{quasirenew2}-\eqref{quasirenew} in Section \ref{sec:prelimann},
that do not necessarily hold if $a\leq 2$).

\subsection{Influence of disorder in the case of summable correlations,
smoothing of the phase transition}

\label{sec:pinsomm}


We now assume that $a>1$, so that correlations are (absolutely) summable. We also assume that $\gU$ is invertible, condition that we comment later, in Remark \ref{rem:gUinfty}.
We show that in presence of disorder, the 
phase transition is always at least of order $2$,
as in the IID  case (see \cite[Th.5.6]{GBbook}), and in the correlated hierarchical model (see \cite[Proposition 3.5]{BThier})

\begin{theorem}
Under Assumption \ref{hyphyp} with $a>1$, and assuming that $\gU$ is invertible, one has that
for every $\alpha\geq0$, for all $\gb>0$ and $h\in\R$
\begin{equation}
 \tf(\gb,h) \leq \frac{1+\ga}{2 \gU_{\infty}\gb^2} \, \left( h-h_c(\gb) \right)_+^2,
\end{equation}
where we defined $\gU_{\infty}:=\left( 1+2\sum_{k\in\bbN} \rho_k \right)\in(0,+\infty)$.
\label{thm:smooth}
\end{theorem}
This stresses the relevance of disorder in the case $\alpha>1/2$,
where the pure model exhibits a phase transition of order $\nu^{\rm pur}:=1\vee 1/\alpha<2$.
Therefore, with summable correlations, we already have identified
a region of the $(\ga, a )$-plane where disorder is relevant:
it corresponds to the relevant disorder region in the IID  case,
as predicted by the Weinrib-Halperin criterion.

\begin{rem}\rm
\label{rem:gUinfty}
The condition that $\gU$ is invertible is a bit delicate, and enables us to get uniform bounds on the eigenvalues of $\gU_l^{-1}$, where $\gU_l$ denotes the restriction of $\gU$ to the first $l$ rows and columns. Indeed,  when $a>1$, $\gU$ is a bounded and invertible operator on the Banach space of sequences of real numbers $(a_n)_{n\in \bbN}$, with finite $\ell_1$-norm $||(a_n)_{n\in \bbN} ||_{\ell_1} := \sum_{n\in\bbN} |a_n| <+\infty$, so that $\gU^{-1}$ is a bounded operator. It tells us that the lowest eigenvalue of $\gU$ is bounded away from $0$, and that the eigenvalues of $\gU_l$ are uniformly bounded away from $0$.

In particular, one has
\begin{equation}
\label{defgUinfty}
\gU_{\infty}:=\lim_{l\to\infty} \frac{\langle \gU_l\ind_l ,\ind_l\rangle}{\langle \ind_l ,\ind_l\rangle}
=1+2\sum_{k\in\bbN} \rho_k>0,
\end{equation}
where $\langle \cdot,\cdot \rangle$ denotes the usual Euclidean scalar product, and $\mathbf{1}_l$ is the vector
constituted of $l$ $1$s and then of $0$s.
Note that $\gU_{\infty}=1$ in the IID case. Also, $\gU_{\infty}$ is an increasing function of the correlations, and it becomes infinite when correlations are no longer summable.
Interestingly, $\gU_{\infty}$ is also related to the relative entropy of two translated Gaussian vectors: Lemma \ref{lem:entropy1} gives that $\langle \gU_l^{-1}\ind_l ,\ind_l\rangle \stackrel{l\to\infty }{\sim} (\gU_{\infty})^{-1}l$. The assumption that $\gU^{-1}$ is a bounded operator plays an important role in the proof of that Lemma.

A simple case when $\gU$ is invertible is when $1=\rho_0 > 2 \sum_{k\in\bbN} |\rho_k|$: it is then diagonally dominant. More generally, one has to consider the Laurent series of the Toeplitz matrix $\gU$,
$f(\lambda) = 1+2\sum_{k\in\bbN} \rho_n \cos (\lambda n)$ (we used that $\rho_0=1$).
Then, the fundamental eigenvalue distribution theorem of Szeg\"o \cite[Ch. 5]{Szego} tells that the Toeplitz operator $\gU$ is invertible if and only if $\min_{\lambda\in[0,2\pi]} f(\lambda)>0$. Note that one recovers the diagonally dominant condition as a consequence of Szeg\"o's theorem. 
\end{rem}

\subsection{The effect of non-summable correlations}
\label{sec:pinnonsomm}

If $a\leq 1$, then correlations are not summable, $\sum_{k\in\bbN} \rho_k=+\infty$, and the annealed model is actually ill-defined.
Indeed, imposing renewal points at every site in $\{1, \ldots,N\}$
in the annealed partition function,
one ends up with the bound
$Z_{N,h}^{\a,\gU} \geq K(1)^N \exp\big( N(h+\gb^2/2+ \gb^2\sum_{k=1}^N \rho_k) \big),$
so that
$\frac{1}{N} \log Z_{N,h}^{\a,\gU} \geq \log \K(1) +\gb^2/2+h+ \gb^2\sum_{k=1}^N \rho_k$.
Letting $N$ go to infinity, we see that the annealed free energy is infinite.

But, when $a<1$, not only the annealed free energy is ill-defined:
we also prove that the quenched free energy is strictly positive for
every value of $h\in \R$: the disordered system does not have a
localization/delocalization phase transition and is always localized, as found in the hierarchical case \cite[Theorem 3.7]{BThier}.
\begin{theorem}
\label{thm:xi<1}
Under Assumption \ref{hyphyp} with $a>1$, if in addition the correlations are non-negative ($\rho_k\geq 0$ for all $k\in\bbN$), one has that $\tf(\gb,h)>0$ for every $\beta>0,h\in\mathbb R$, so
that $h_c^{\q}(\gb)=-\infty$. There exists some constant $c_2>0$ such that
for all $h\leq -1$ and $\gb>0$
\begin{equation}
\label{eq:xi<1}
 \tf(\gb,h)\geq  \exp\left(- c_2|h|\left(|h|/\gb^{2}\right)^{1/(1- a )}\right).
\end{equation}
\end{theorem}
The non-negativity condition for the correlations is only technical (and appears in the proof of Lemma \ref{lem:entropy2}), and we believe that the same result should be true with a more general correlation structure.

This shows that the phase transition disappears when correlations
are too strong. It provides an example where strongly correlated disorder always modifies (in an extreme way) the behavior of the system, for every value of the renewal parameter $\alpha$. However the fact that $h_c^{\q}(\gb)=-\infty$ does not allow us to study sharply how the phase transition is modified by the presence of disorder, and therefore we cannot verify nor contradict the Weinrib-Halperin prediction.
This phenomenon comes from the appearance of large, frequent, and arbitrarily favorable regions in the environment. This is the mark of the appearance of {\sl infinite disorder}, and is studied in depth in \cite{Bstrong}.

\medskip

We now have a clearer picture of the behavior of the disordered system,
and of its dependence on the strength of the correlations,
that we collect in Figure \ref{fig:zones}, to be compared with \cite[Fig. 1]{BThier} for the hierarchical pinning model.

\begin{figure}[htbp]
\centerline{
\psfrag{0}{$0$}
\psfrag{z}{$ a $}
\psfrag{z1}{$ a =1$}
\psfrag{z2}{$ a =2$}
\psfrag{ga}{$\ga$}
\psfrag{ga12}{$\ga=1/2$}
\psfrag{ga1}{$\ga=1$}
\psfrag{nonsommable}{{\bf No phase transition}}
\psfrag{WHcrit}{\small Weinrib-Halperin criterion}
\psfrag{WHline}{$ a  =2(1\wedge\ga)$}
\psfrag{?}{\small$\nu^{\rm pur}\stackrel{??}{<}\nu^{\a}$}
\psfrag{Rel}{\small Relevant Disorder}
\psfrag{Irr}{\small ?\!\!  Irrelevant Disorder\!\! ?}
\psfrag{nuannrel}{\small $\nu^{\a}=\nu^{\rm pur}$}
\psfrag{nurel}{\small $\nu^{\rm pur}<2\leq \nu^{\rm que}$}
\psfrag{nuannirr}{\small $\nu^{\a}=\nu^{\rm pur}$}
\psfrag{nuanndur}{\small$\nu^{\a}\stackrel{??}{>}\nu^{\rm pur}$}
\psfrag{nuann}{\small$\nu^{\a}\stackrel{??}{=}\nu^{\rm pur}$}
\psfig{file=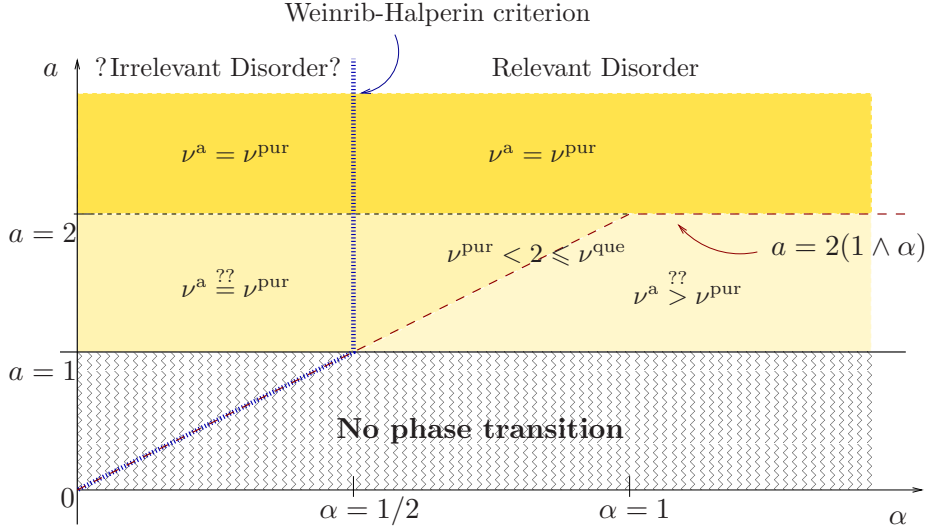,width=4.8in}}
  \begin{center}
    \caption{\label{fig:zones}
      Overview of the annealed behavior and of disorder relevance/ir\-re\-le\-vance
      in the $(\ga, a )$-plane, in analogy with \cite[Fig. 1]{BThier}.
      In the region $ a <1$ (non-summable correlations),
      the annealed model is not well-defined, and
      there is no phase transition for the disordered system (Theorem \ref{thm:xi<1}).
      In the region $ a >1$, the annealed model is well-defined,
      and Theorem \ref{thm:expo_ann} shows that the annealed critical behavior
      is the same that the pure one if $ a >2$. Theorem \ref{thm:smooth}
      shows that disorder is relevant for $\ga>1/2$,
      but we still have no proof of disorder irrelevance for $\ga<1/2$,
      that we believe to hold according to the physicists' predictions.
      For $ a \in(1,2)$, the results in the hierarchical case
      indicate that the annealed critical exponent $\nu^{\a}$
      should be equal to the pure one $\nu^{\rm pur}$
      if $ a \nu^{\rm pur}>2$ (\textit{i.e.}\  $ a >2 (1\wedge\ga)$),
      and that it should be strictly larger
      if $ a \nu^{\rm pur}<2$. }
  \end{center}
\end{figure}

\section{The annealed model}
\label{sec:anncorrel}

\subsection{Preliminary observations on the annealed partition function}
\label{sec:prelimann}

We now give the reason why the condition $a>2$
simplifies the analysis of the annealed system.
Given two arbitrary disjoint blocks $B_1$ and $B_2$, the contribution to the Hamiltonian \eqref{defZann} of these two blocks can be divided into:
\begin{itemize}
 \item two {\sl internal} contributions
$(\gb^{2}/2+h) \sum_{i\in B_s}\gd_i + \gb^2\sum_{i,j\in B_s,i<j} \gd_i\gd_j \rho_{|i-j|}$
for $s=1,2$,

\item an {\sl interaction} contribution $\gb^2\sum_{i\in B_1,j\in B_2} \gd_i\gd_j \rho_{|i-j|}$.
\end{itemize}
We also refer to the latter term as the correlation term.
Then we can use
uniform bounds to control the interactions between $B_1$ and $B_2$,
since there are at most $k$ points at distance $k$ between $B_1$ and $B_2$:
\begin{equation}
\label{eq:boundcor}
-\sum_{k=1}^{\infty} k |\rho_k|\leq
 \sum_{k=1}^{\infty} \rho_k \sumtwo{i\in B_1,j\in B_2}{|i-j|=k} \gd_i \gd_j \leq \sum_{k=1}^{\infty} k |\rho_k| .
\end{equation}

Thanks to this remark, if $a>2$, then $\sum k|\rho_k|<\infty$, and
we have a "quasi super-multiplicativity'' property (super-multiplicativity
would hold if all of the $\rho_k$ were non-negative):
for any $N\geq1$ and $0\leq k\leq N$, one has
\begin{equation}
 \Zna \geq e^{-\gb^2 \sum k|\rho_k|} Z_{k,h}^{\a} Z_{N-k,h}^{\a}.
\label{eq:surmult}
\end{equation}
We also get the two following bounds, which can be seen as substitutes
for the renewal property (property that we do not have in our annealed system because
of the two-body $\gd_i\gd_j$ term).
Decomposing according to the last renewal before 
some integer $M\in[0, N]$,
and the first after it, one gets
\begin{equation}
 \Zna \geq  
   \sum_{i=0}^{M} \sum_{j=M+1}^{N} e^{-\gb^2 \sum k|\rho_k|}Z_{i,h}^{\a}\,
 \K(j-i)e^{\gb^2/2+h-\gb^2\sum |\rho_k|}\, Z_{N-j,h}^{\a},
\end{equation}
and
\begin{equation}
 \Zna \leq 
   \sum_{i=0}^{M} \sum_{j=M+1}^{N} e^{\gb^2 \sum k|\rho_k|} Z_{i,h}^{\a}\,
 \K(j-i)e^{\gb^2/2+h+\gb^2\sum |\rho_k|}\, Z_{N-j,h}^{\a}.
\end{equation}
Note that the terms $e^{\gb^2/2+h-\gb^2\sum |\rho_k| }$ and $e^{\gb^2/2+h+\gb^2\sum |\rho_k|}$ come from bounding
uniformly the contribution of the point $j$ to the partition function (note that $\sum |\rho_k|<+\infty$ because $a>2$).
If we write $h=\hca+u$, and using that $\hca$
is of order $\gb^2$ (see \eqref{asymppointann}),
we get a constant $c>0$ such that
\begin{equation}
e^{-c \gb^2} e^u \sum_{i=0}^{M} \sum_{j=M+1}^{N} Z_{i,h}^{\a} \K(j-i)Z_{N-j,h}^{\a}
    \leq \Zna \leq e^{c \gb^2} e^u\sum_{i=0}^{M} \sum_{j=M+1}^{N} Z_{i,h}^{\a} \K(j-i) Z_{N-j,h}^{\a}.
\label{quasirenew2}
\end{equation}
Note that one has also uniform bounds for $u\in[-1,1]$ (we are interested in the critical behavior, \textit{i.e.}\ for $u $ close to $0$):
one replaces the constant $e^{c \gb^2 }e^u$ by $C_1:=e^{c \gb^2 +1}$,
and the constant $e^{-c \gb^2 }e^u$ by $C_1^{-1}$.

In a general way, for any indexes $0=i_0<i_1<i_2<\cdots<i_m=N$, we also get
\begin{equation}
\left(e^{-c \gb^2} e^u\right)^m \prod_{k=1}^m Z_{i_k-i_{k-1},h}^{\a}
    \leq \E\left[ \prod_{k=1}^m \gd_{i_k} e^{H_{N,h}^{\a}}\right]
       \leq \left(e^{c \gb^2} e^u\right)^m  \prod_{k=1}^m Z_{i_k-i_{k-1},h}^{\a} .
\label{quasirenew}
\end{equation} 
When $\gb$
is small, \eqref{quasirenew2}-\eqref{quasirenew} are close to the renewal equation verified by $Z_{N,h}^{\rm pur}$
which is the same as \eqref{quasirenew2}-\eqref{quasirenew} with $\gb=0$.
In the sequel, we refer to \eqref{quasirenew2}-\eqref{quasirenew} as the \emph{quasi-renewal property}.
We can actually show Theorem \ref{thm:expo_ann} provided that these inequalities
hold. Therefore if one is able to get \eqref{quasirenew2}-\eqref{quasirenew}
with a weaker condition than $\sum k|\rho_k|<\infty$ (which could be
$ a >2(\ga\wedge1)$, as the comparison with the hierarchical model suggests, see Section \ref{sec:comparhier}),
such a theorem would follow.

\subsection{The annealed critical behavior}
\label{sec:behavann}

\subsubsection{On the resolution of the homogeneous model}
Our proof of Theorem \ref{thm:expo_ann}
is inspired from the following proposition.
\begin{proposition}
\label{prop:defFbis}
The homogeneous free energy $\tf(h)$ is the only solution
of the equation (in $b$)
\begin{equation}
\label{defb_bis}
\hP(b):=\sum_{n\geq 0} e^{-bn} \P(n\in\tau)=\frac{1}{1-e^{-h}}
\end{equation}
if such a solution exists, and $\tf(h)=0$ otherwise.
Thanks to \eqref{defb_bis}, one is in particular able to recover
Proposition \ref{comphomo}.
\end{proposition}

The proof of such a result is standard, and we refer to \cite[Proposition 1.1]{GBbook}, which gives a slightly different form: $\tf(h)$ is the only solution of the equation
$\hat \K(b):=\sum_{n\in\N} e^{-bn} \K(n)=e^{-h} $. One recovers Proposition \ref{prop:defFbis} observing that $\hP(b) = 1+ \hat\K(b) \hP(b)$.


\subsubsection{Proof of Theorem \ref{thm:expo_ann}}

We now drop the superscript $\gU$ in $Z_{N,h}^{\a,\gU}$,
and write $\hca$ instead of $h_c^{\a,\gU}(\gb)$,
to keep notations simple.

The essential tool is to use the quasi-renewal property \eqref{quasirenew2}-\eqref{quasirenew}  to prove that the Laplace transform of $Z_{n,\hca}^{\a}$ is of the same order as $\hP(\gl)$, the Laplace transform of $Z_{N,h=0}^{\rm pur} =\bP(n\in\tau)$. Then, one would to be able to apply the same idea as in Proposition \ref{prop:defFbis}.
The following proposition indeed proves that statement.
\begin{proposition}
\label{prop:laplace}
We assume that the quasi-renewal property \eqref{quasirenew2}-\eqref{quasirenew} holds.
Define for all $\lambda>0$ $\hat Z_{\hca}(\lambda):=\sum_{n=0}^{\infty} e^{-\lambda n} Z_{n,\hca}^{\a}$.
Then there exists a constant $c_3>0$, such that for every $0<\gl\leq 1$ one has
\begin{equation}
 c_3^{-1} \hP (\gl ) \leq \hat Z_{\hca}(\lambda) \leq   c_3 \hP(\gl).
\end{equation}
\end{proposition}

We remark that in \cite{BThier}, the key for the study of the disordered system via annealed techniques is a sharp control of the annealed polymer measure at its critical point (even though the exact value of this critical point is not known). In the present case, since there is no iterative structure for the partition function, there are many technicalities that are harder to deal with. We however have results in this direction, such as Propositions \ref{prop:laplace}-\ref{prop:Pncbis}, that are the first step towards proving that the Harris criterion holds if $\sum k|\rho_k|<\infty$, in terms of critical point shifts. We do not develop the analysis in this direction, which is still open and would require a stronger knowledge of the annealed system.

\begin{proof}[Proof of Theorem \ref{thm:expo_ann} given Proposition \ref{prop:laplace}]
Recall that we define $u:=h-\hca$, so that
we only work with $u>0$, $u\in[0,1]$, as we already know that for $u\leq 0$, $\tf^{\a}(\gb,u)=0=\tf(u)$.
We use the following binomial expansion
\begin{equation}
 e^{u\sum_{n=1}^{N} \gd_n} \gd_N = (1+e^u-1)^{\sum_{n=1}^{N-1} \gd_n} e^u \gd_N    = e^u\sum_{m=0}^{N-1} (e^u-1)^m \sum_{0<i_1<\ldots < i_m\leq N-1 } \gd_{i_1} \ldots \gd_{i_m} \gd_N,
\end{equation}
to get that
\begin{equation}
\label{eq:expand_Zna}
   \Zna = \bE\left[ e^{u\sum_{n=1}^N \gd_n}  e^{H_{n,\hca}^a} \gd_N\right]
        = \frac{e^u}{e^u-1}\sum_{m=1}^N \sum_{0<i_1<\ldots < i_m = N } (e^u-1)^m
        \bE\left[ \gd_{i_1}\ldots \gd_{i_m}  e^{H_{N,\hca}^{\a}}\right].
\end{equation}
Note that as there is no renewal structure for $\E\left[ \cdot\  e^{H_{N,\hca}^{\a}}\right]$,
one cannot factorize the quantity $\E\left[ \gd_{i_1}\ldots \gd_{i_m} e^{H_{N,\hca}^{\a}} \right]$ easily.
However, since we have the quasi-renewal property \eqref{quasirenew},
we get the two following bounds, valid for any $m\in\bbN$ and subsequence
$0<i_1<\ldots < i_m = N$,  uniformly for $u\in[0,1]$:
\begin{equation}
(C_1^{-1})^{m} \prod_{k=1}^m Z_{i_k-i_{k-1},h_c^{\a}}^{\a}  
   \leq \bE\left[ \gd_{i_1}\ldots \gd_{i_m} \gd_N e^{H_{n,\hca}^a}\right]
     \leq  (C_1)^m \prod_{k=1}^m Z_{i_k-i_{k-1},h_c^{\a}}^{\a} ,
\label{eq:boundpnc}
\end{equation}
where $C_1:=e^{c\gb+1}$ is defined in Section \ref{sec:prelimann}.
Now, we define
\begin{equation}
 \begin{split}
 \bar Z_{N,h}^{\a} &:= \frac{e^u}{e^u-1} \sum_{m=1}^N \left(C_1^{-1}(e^u-1)\right)^m\sum_{0<i_1<\ldots < i_m = N }
            \prod_{k=1}^m Z_{i_k-i_{k-1},h_c^{\a}}^{\a}   \\
\text{and }\quad \tilde Z_{N,h}^{\a} &:= \frac{e^u}{e^u-1} 
\sum_{m=1}^N \left( C_1(e^u-1) \right)^m\sum_{0<i_1<\ldots < i_m = N }
             \prod_{k=1}^m Z_{i_k-i_{k-1},h_c^{\a}}^{\a} ,  
 \end{split}
\end{equation} 
so that $\bar Z_{N,h}^{\a} \leq \Zna \leq \tilde Z_{N,h}^{\a}$.
For $u>0$, we can define $\bar b >0$ and $\tilde b>0$ such that
\begin{equation}
   \hat Z_{h_c^{\a}} (\, \bar b \, ) = C_1(1-e^{-u})^{-1}, \indent
   \text{and } \indent \hat Z_{h_c^{\a}} (\, \tilde b \, ) = C_1^{-1}(1-e^{-u})^{-1},
\label{def_b1b2}
\end{equation}
if the equations have a solution and otherwise set $\bar b=0$, or $\tilde b=0$. Then, if one defines $\K^{(\bar b)}(n):= C_1^{-1} (1-e^{-u}) e^{-\bar b n}\Zna$ for $n\geq 1$, one verifies that $K^{(\bar b)}(\cdot)$ is the inter-arrival distribution of a positive recurrent renewal $\tau^{(\bar b)}$. Moreover, with this definition, $\bar Z_{N,h}^{\a} = \frac{e^u}{e^u-1} e^{\bar b N} \bP^{(\bar b)} (n\in\tau^{(\bar b)})$, and one gets that $\lim \frac{1}{N}\log \bar Z_{N,h}^{\a} = \bar b$. Similarly, one has that $\lim \frac{1}{N}\log \tilde Z_{N,h}^{\a} = \tilde b$.

Then, one gets that $ \bar b \leq \tf^{\a}(\gb,\hca+u) \leq \tilde b$,
from the fact that $\bar Z_{N,h}^{\a} \leq \Zna \leq \tilde Z_{N,h}^{\a}$.
Using that $\hP(\cdot)$ is decreasing, one therefore has that $\hP (\tilde b) \leq \hP(\tf^{\a}(\gb,\hca+u)) \leq \hP(\bar b)$.
The definitions \eqref{def_b1b2}, combined with Proposition \ref{prop:laplace}, gives that for every $u>0$ such that $\bar b\leq 1$ one has
\begin{equation}
( c_3 C_1)^{-1}  (1-e^{-u})^{-1} \leq \hP(\tf^{\a}(\gb,\hca+u)) \leq  c_3C_1 (1-e^{-u})^{-1}.
\end{equation}
We finally have that for $u\geq 0$ small enough, there are two constants $c$ and $c'$ such that
\begin{equation}
  (1-e^{-cu})^{-1}\leq \hP(\tf^{\a}(\gb,\hca+u)) \leq (1-e^{-c'u})^{-1}.
\end{equation}
Applying the inverse of $\hP$ (which is also decreasing), one gets the result from
the fact that $\tf(u) =\hP((1-e^{-u})^{-1})$ (see \eqref{defb_bis}).
\end{proof}

\subsection{Proof of Proposition \ref{prop:laplace}}

Let us first note that, thanks to Kamarata's Tauberian Theorem  \cite[Theorem 1.7.1]{Bingham}, the asymptotic behavior of the Laplace transform of $\Zna$ is directly related to that of $\sum_{n=0}^{N} Z_{n,h_c^{\a}}^{\a}$.
We also stress that $\sum_{n=0}^{N} \bP(n\in\tau) \stackrel{N\to\infty}{\sim} \tilde\varphi(N) N^{\min(\alpha,1)}$, where $\bar\varphi(\cdot)$ is a well determined slowly varying function, see \cite[Theorems 8.7.3 and 8.7.5]{Bingham}: for example, $\tilde\varphi(n)= \frac{\sin(\pi \ga)}{\pi}\varphi(n)^{-1}$ if $\ga\in(0,1)$, and $\tilde\varphi(n)=\bE[\tau_1]^{-1}$ if $\ga>1$, the cases $\alpha=0$ and $\alpha=1$ requiring more care.

To avoid too many technicalities, we will focus only on the cases $\ga\in(0,1)$ and $\ga>1$, the cases $\ga=0$ and $\ga=1$ following the same proof.
We therefore only need to prove that there exists a constant $c_4$ so that, for all $N\in\bbN$, 
\begin{gather}
\label{sumstobeproved}
c_4^{-1}\varphi(N)^{-1} N^{\ga} \leq \sum_{n=0}^{N} Z_{n,h_c^{\a}}^{\a} \leq c_4 \varphi(N)^{-1}N^{\ga}   \quad \text{ if } \ga\in(0,1) ,\\
 \text{and }\quad \quad c_4^{-1} N\leq \sum_{n=0}^{N} Z_{n,h_c^{\a}}^{\a} \leq c_4 N\quad\quad \quad\quad \text{ if } \ga>1.
\end{gather}
In terms of Laplace transforms, one has to show that there exists a constant $c_5$ so that, for all  for all $\lambda\in(0,1)$
\begin{gather}
\label{laplacetobeproved}
c_5^{-1}\varphi(1/\gl)^{-1} \gl^{-\ga}  \leq\hat Z_{n,h_c^{\a}}^{\a} (\lambda)\leq c_5\, \varphi(1/\gl)^{-1} \gl^{-\ga}   \quad \text{ if } \ga\in(0,1) ,\\
 \text{and }\quad \quad c_5^{-1} \lambda^{-1}\leq \hat Z_{n,h_c^{\a}}^{\a} (\gl ) \leq c_5 \gl^{-1}\quad \quad \quad \text{ if } \ga>1.
\end{gather}
The behavior of the Laplace transform $\hP(\gl)$ can be found using \eqref{sumstobeproved}, together with \cite[Theorem 1.7.1]{Bingham}. Note that the lower bounds (resp. the upper bounds) in \eqref{sumstobeproved} correspond to the lower bounds (resp. the upper bounds) in \eqref{laplacetobeproved}.

\smallskip
Let us first prove a preliminary result that will be useful, both in the case $\ga\in(0,1)$,
and in the case $\ga>1$.
\begin{claim}
 \label{claim:Zhcpetit}
For every $\ga>0$, if the quasi-renewal property \eqref{quasirenew2}-\eqref{quasirenew} holds,
then for all $N\in\N$ one has $Z_{N,\hca}^{\a}\leq C_1$, where $C_1=e^{c\gb^2+1}$
is defined above. 
\end{claim}
Indeed, the l.h.s. inequality in \eqref{quasirenew} yields that for all $u\in[-1,1]$,
one has
$$C_1^{-1} Z_{M+N,h}^{\a} \geq (C_1^{-1} Z_{N,h}^{\a}) (C_1^{-1} Z_{M,h}^{\a})$$
 for all $M,N\geq 0$.
Therefore one gets that if $C_1^{-1} Z_{n_0,h}^{\a}>1$ for some $n_0$, then the partition function
grows exponentially, and $\tf(\gb,h)>0$.
This gives directly that $C_1^{-1} Z_{N,\hca}^{\a}\leq 1$ for all $N\in\N$.
\qed

Let us focus first on the case $\ga\in(0,1)$, since Proposition \ref{prop:Pncbis} gives a better result in the case $\ga>1$, and prove \eqref{sumstobeproved}.

\smallskip
\textbf{Upper bound.}
We prove the following Lemma 
\begin{lemma}
For $\ga\in(0,1)$, there exists a constant $C_0>0$ such that for any $N\geq 1$
\begin{equation}
 \sum_{n=0}^{N} Z_{n,h_c^{\a}}^{\a} \leq C_0 \varphi(N)^{-1}N^{\ga}.
\end{equation}
\label{lem:sumZnc}
\end{lemma}

\begin{proof}
 If the Lemma were not true, then for any constant $A>0$ arbitrarily large,
 there would exist some $n_0\geq 1$ such that
\begin{equation}
\sum_{n=0}^{n_0} Z_{n,h_c^{\a}}^{\a} \geq A \varphi(n_0)^{-1} n_0^{\ga}.
\label{eq:sumZnc}
\end{equation}
But in this case, using the l.h.s. inequality of \eqref{quasirenew2}, we get
for any $ 2n_0\leq p\leq 4n_0$
\begin{multline}
Z_{p,h_c^{\a}}^{\a} 
  \geq  C_1^{-1}
\sum_{i=0}^{\floor{p/2} } \sum_{j=\floor{p/2}+1}^{p} Z_{i,h_c^{\a}}^{\a} \K(j-i)Z_{p-j,h_c^{\a}}^{\a} \\
   \geq  C_1^{-1} \left(\sum_{i=0}^{n_0} \sum_{j=p-n_0}^{p} Z_{i,h_c^{\a}}^{\a} Z_{p-j,h_c^{\a}}^{\a} \right)
                \min_{n\leq p} \K(n)
             \geq  C_1^{-1} A^2 \varphi(n_0)^{-2} n_0^{2\ga}  \min_{n\leq 4n_0} \K(n),
\end{multline}
where we restricted the sum to $i$ and $p-j$ smaller than $n_0$,
to be able to use the inequality \eqref{eq:sumZnc}.
On the other hand, with the assumption that $\K(n)\sim \varphi(n) n^{-(1+\ga)}$,
there exists a constant $c>0$ (not depending on $n_0$)
such that one has that $\min_{n\leq 4n_0} \K(n) \geq c\varphi(n_0) n_0^{-(1+\ga)}$.
And thus for any $2 n_0 \leq p \leq 4 n_0$ one has that
$$ Z_{p,h_c^{\a}}^{\a} \geq c' A^2 \varphi(n_0)^{-1} n_0^{\ga-1}.$$
 Then, summing over $p$, we get an inequality similar to \eqref{eq:sumZnc}:
\begin{equation}
\sum_{p=0}^{4n_0} Z_{p,h_c^{\a}}^{\a} \geq \sum_{p=2n_0}^{4n_0} Z_{p,h_c^{\a}}^{\a} \geq c'' A^2 \varphi(n_0)^{-1} n_0^{\ga} =: \bar{c}A^2 \varphi(4n_0)^{-1} (4n_0)^{\ga}.
\end{equation}
Now, we are able to repeat this argument with $n_0$ replaced with $4n_0$ and
$A$ with $\bar{c}A^2$.
By induction, we finally have for any $k\geq 0$
\begin{equation}
 \sum_{n=0}^{4^k n_0} Z_{n,h_c^{\a}}^{\a} \geq  (\bar{c})^{2^k-1}A^{2^k} \varphi(4^k n_0)^{-1} (4^k n_0)^{\ga}.
\label{eq:sumZnc2}
\end{equation}
To find a contradiction, we choose $A>(\bar{c})^{-1}$, so that $(\bar{c})^{2^k-1}A^{2^k}\geq \gga^{2^k}$
with $\gga>1$. Now, we can choose $k\in\N$ such that $\gga^{2^k} \varphi(4^k n_0)^{-1} (4^k n_0)^{\ga-1}\geq 2C_1$
($C_1$ being the constant in Claim \ref{claim:Zhcpetit}).
Thanks to \eqref{eq:sumZnc2}, we get that at least one of the terms
$Z_{n,h_c^{\a}}^{\a}$ for $ n\leq 4^k n_0$ is bigger than $(4^k n_0)^{\ga-1} \gga^{2^k}\geq 2C_1$,
which contradicts the Claim \ref{claim:Zhcpetit}.
\end{proof}

\textbf{Lower Bound.} We use the following Lemma
\begin{lemma}
If $\alpha\in(0,1)$, there exists some $\eta>0$, such that
if for some $n_0\geq 1$ one has
\begin{equation}
\sum_{i=0}^{n_0} Z_{i,h}^{\a} \sum_{j=n_0}^{\infty}  \K(j-i) \leq \eta \indent 
  \text{ and } \indent \sum_{i=0}^{n_0} Z_{i,h}^{\a} \leq \eta \varphi(n_0)^{-1}  n_0^{\ga} ,
\end{equation} 
then $\tf^{\a}(\gb,h)=0$.
\label{lem:finite_deloc}
\end{lemma}
This Lemma comes easily from \cite[Lemma 5.2]{GLT08} where the case $\ga=1/2$ was considered, and gives a finite-size criterion for delocalization. It comes from cutting the system into blocks of size $n_0$, and then using a coarse-graining argument in order to reduce the analysis to finite-size estimates (on segments of size $\leq n_0$).
It is therefore not difficult to extend it to every $\ga \in (0,1)$, in particular thanks to the quasi-renewal property \eqref{quasirenew2}-\eqref{quasirenew}, that allows us to proceed to the coarse-graining decomposition of the system.

From this Lemma, one deduces that at $h=h_c^{\a}$, for all $n\in\N$ one has
\begin{eqnarray}
   & \sum_{i=1}^{n} \sum_{j=n}^{\infty} Z_{i,h_c^{\a}}^{\a} \K(j-i) \geq \frac{\eta}{2} &\label{eq:condupper1}\\
 & \text{or }\sum_{i=1}^{n} Z_{i,h_c^{\a}}^{\a} \geq \frac{\eta}{2} \varphi(n)^{-1} n^{\ga}.&  \label{eq:condupper2}
\end{eqnarray} 
Indeed, otherwise, one could find some $n_0\geq 0$ such that both of these assumptions fail,
and then one picks some $\gep>0$ such that $Z_{n_0,h_c^{\a}+\gep}^{\a}$ verifies the conditions of
Lemma \ref{lem:finite_deloc}, so that $\tf^{\a}(\gb,\hca+\gep)=0$.
This contradicts the definition of $h_c^{\a}$.

\smallskip
We now try to deduce directly the behavior of $\hat Z_{h_c^{\a}}^{\a}(\lambda)$ from
\eqref{eq:condupper1}-\eqref{eq:condupper2} (it turns out to be easier). We define the sets
\begin{equation}
 \begin{split}
  E_1 &:= \left\{ n\geq 0,  \text{ such that } \eqref{eq:condupper1}\text{ holds} \right\}, \\
  E_2 &:= \left\{ n\geq 0,  \text{ such that } \eqref{eq:condupper2}\text{ holds} \right\}.
 \end{split}
\end{equation} 
Thanks to \eqref{eq:condupper1}-\eqref{eq:condupper2}, one knows that, for every $k\in\bbN$, either $|E_1\cap[0,k]|\geq k/2$ or $|E_2\cap[0,k]|\geq k/2$.
Let us fix $\lambda\in(0,1)$, and $k_{\lambda}:=\floor{1/\gl}$.

\smallskip
{(1) If $|E_1\cap[0,k_{\lambda}]|\geq k_{\lambda}/2$.}
For $\gl>0$, we define $f(\gl)=\sum_{n=0}^{\infty} e^{-\gl n} \varphi(n)(n+1)^{-\ga}$.
We know that $f(\gl)\stackrel{\gl\downarrow 0}{\sim} cst. \varphi(1/\gl) \gl^{\ga-1}$ thanks to \cite[Th.1.7.1]{Bingham}.
Then, using the assumption on $\K(\cdot)$ to find some constant $c>0$
such that for all $i\leq n$ one has $\sum_{j=n}^{\infty} \K(j-i) \leq c\varphi(n-i) (n+1-i)^{-\ga}$, one gets
\begin{multline}
 \hat Z_{h_c^{\a}}(\gl) f(\gl) = 
  \sum_{n=0}^{\infty} e^{-\gl n} \sum_{i=1}^n Z_{i,h_c^{\a}}^{\a} \varphi(n-i)(n+1-i)^{-\ga} 
   \geq \sum_{n=0}^{\infty} c^{-1}e^{-\gl n} \sum_{i=1}^n Z_{i,h_c^{\a}}^{\a} \sum_{j=n}^{\infty} \K(j-i)  \\
\geq c^{-1}\eta/2 \sum_{n\in\E_1} e^{-\gl n}  \geq  c^{-1} e^{-1}\eta/2\left| E_1 \cap [0,k_{\lambda}]  \right| \geq c^{-1} e^{-1} \frac{\eta}{4}k_{\lambda}  ,
\end{multline} 
where in the second inequality we used the definition of $E_1$, and then we cut the sum at $k_{\lambda}$.
Thus we get from our estimate on $f(\gl)$, that there exists a constant $c$, so that for any fixed $\gl\in(0, 1)$ (recall that $k_{\lambda}=\floor{1/\gl}$), if $|E_1\cap[0,k_{\lambda}]|\geq k_{\lambda}/2$, then
\begin{equation}
 \hat Z_{h_c^{\a}} \geq c\,  \varphi(1/\gl)^{-1} \gl^{-\ga}.
\label{eq:hZlow1}
\end{equation}
\smallskip

{(2) If $|E_2\cap[0,k_{\lambda}]|\geq k_{\lambda}/2$.}
Then, using the definition of $E_2$ and the notation $k_2:=\max (E_2\cap \{1,\ldots, \floor{1/\gl}\})$, one has
\begin{equation}
 \hat Z_{h_c^{\a}}(\gl) \geq
e^{-1}\sum_{i=0}^{\max(E_2\cap \{1,\ldots, \floor{1/\gl}\})} Z_{i,h_c^{\a}}^{\a} \geq e^{-1}\frac{\eta}{2} \varphi(k_2)^{-1} k_2^{\ga}.
\end{equation}
Note that $k_2\in [k_{\gl}/2,k_{\gl}]$ if $|E_2\cap[0,k_{\lambda}]|\geq k_{\lambda}/2$. Therefore, there exists a constant $c'$, so that for any fixed $\gl\in(0, 1)$ (recall that $k_{\lambda}=\floor{1/\gl}$), if $|E_2\cap[0,k_{\lambda}]|\geq k_{\lambda}/2$, then
\begin{equation}
 \hat Z_{h_c^{\a}} \geq c' \, \varphi(1/\gl)^{-1} \gl^{-\ga}.
\label{eq:hZlow2}
\end{equation}

Then, combining \eqref{eq:hZlow1} and \eqref{eq:hZlow2},
we get our  $\hat Z_{h_c^{\a}}^{\a}(\gl) \geq \min(c,c')\, \varphi(1/\gl)^{-1} \gl^{-\ga}$ for all $\gl\in(0, 1)$. 
\qed

\subsubsection{Improvement of Proposition \ref{prop:laplace} in the case $\ga>1$.}
In this case, we can estimate $Z_{N,h_c^{\a}}^{\a}$ more precisely, and estimate not only
the Laplace transform of $Z_{N,\hca}^{\a}$ (cf. Proposition \ref{prop:laplace}), but $Z_{N,\hca}^{\a}$ itself, similarly to \cite[Proposition 3.2]{BThier}.

\begin{proposition} Let $\ga>1$.
Assume that the quasi-renewal property \eqref{quasirenew2}-\eqref{quasirenew} holds.
Then there exists a constants $c_6$ such that, for
any $N\geq2$ and any sequence of indexes $1\leq i_1\leq i_2\leq \ldots \leq i_m = N$ with $m\geq 1$, we have
\begin{equation}
 (c_6^{-1})^{m} \bE(\gd_{i_1} \ldots \gd_{i_m}) \leq \bE\left[ \gd_{i_1}\ldots \gd_{i_m}e^{H_{N,\hca}^a}\right]
  \leq (c_6)^{m} \bE(\gd_{i_1} \ldots \gd_{i_m}).
\end{equation}
In particular, if $m=1$ one has that $c_6^{-1} \bP(N\in\tau) \leq Z_{N,h_c^{\a}}^{\a} \leq c_6 \bP(N\in\tau)$. 
\label{prop:Pncbis}
\end{proposition}

This Proposition tells that the annealed polymer measure at the critical point is ``close'' to the renewal measure $\bP$, so that the behavior of the annealed model is very close to the one of the homogeneous model. 

\smallskip
We have $\bE(\gd_{i_1} \ldots \gd_{i_m}) = \prod_{k=1}^m\bP(i_k-i_{k-1} \in\tau),$
so that recalling \eqref{eq:boundpnc}, we only have to compare $Z_{n,h_c^{\a}}^{\a}$ with $\bP(n\in\tau)$. It is therefore sufficient  to prove that there exists a constant $c$ such that ${c^{-1} {\bP(N\in\tau)} \leq Z_{N,h_c^{\a}}^{\a} }\allowbreak \leq  c \bP(N\in\tau)$.
But for $\ga>1$, we have $\lim_{N\to\infty} \bP(N\in\tau)= \bE[\tau_1]^{-1}$,  so that we only have to show that $Z_{N,h_c^{\a}}^{\a}$ is bounded
away from $0$ and $+\infty$,
which is provided by the following lemma.
\begin{lemma}
If \eqref{quasirenew2}-\eqref{quasirenew} hold,
and if $\ga>1$, there exists a constant $C_2>0$, such that for all $N\geq 0$
\begin{equation}
  C_2^{-1} \leq Z_{N,h_c^{\a}}^{\a} \leq C_2
\end{equation}
\label{lem:Znc}
\end{lemma}

\begin{proof}
The upper bound is already given by Claim \ref{claim:Zhcpetit}, thanks to quasi super-multiplicativity. For the other bound, we show the following claim.

\begin{claim}
If \eqref{quasirenew2}-\eqref{quasirenew} hold, and if $\ga>1$,
let $\gep>0$ (small) and $A>0$ (large) be fixed according to the conditions
\eqref{eq:condAeps1}-\eqref{eq:condAeps2} below.
Then for every $N\geq 0$, there exists some $n_1\in [N-A,N]$ such that
$Z_{n_1,h_c^{\a}}^{\a}\geq \gep$.
\label{claim:Znc}
\end{claim}
From this Claim and inequality \eqref{quasirenew2} with the choice $M=N-1$, we have
\begin{equation}
\Znc \geq C_1^{-1} \sum_{n=0}^{N-1} Z_{n,h_c^{\a}}^{\a}  \K(N-n) e^{\gb^2/2+h_c^{\a}}
  \geq C' Z_{n_1,h_c^{\a}}^{\a} \K(N-n_1),
\end{equation}
where we only kept the term $n=n_1$ in the sum,
$n_1$ being given by the Claim \ref{claim:Znc}.
We get that for every $N\geq 0$,
\begin{equation}
  \Znc \geq \gep  C'\, \big( \min_{i\leq A} \K(i) \big)e^{\gb^2/2+h_c^{\a}} =: C_2^{-1},
\end{equation}
which ends the proof of Lemma \ref{lem:Znc}.
\end{proof}

Now, we prove the Claim \ref{claim:Znc} by contradiction. The idea
is to prove that if the claim were not true, we can increase a bit the parameter $h$ and still
be in the delocalized phase.

\begin{proof}[Proof of Claim \ref{claim:Znc}]
Let us suppose that the claim is not true.
Then
we can find some $n_0$, such that for any $k\in[n_0-A,n_0]$ one has $Z_{k,h_c^{\a}}^{\a}\leq \gep$.
The integer $n_0$ being fixed,
we choose some $h>h_c^{\a}$ close enough to $ \hca$ such that for this $n_0$,
we have (recall $Z_{n,h_c^{\a}}^{\a}\leq C_1$)
\begin{eqnarray}
     & Z_{n,h}^{\a}\leq 2C_1 & \text{ for all }\ n\leq n_0 , \label{cond1:propag} \\
 \text{and} & Z_{k,h}^{\a}\leq 2\gep & \text{ for all }\ k\in[n_0-A,n_0]  \label{cond2:propag}.
\end{eqnarray}
We will now see
that the properties \eqref{cond1:propag}-\eqref{cond2:propag} are kept
when we consider bigger systems: we show that we have
$Z_{n,h}^{\a}\leq 2 C_1$ for all $n\leq 2n_0$ , and
$Z_{k,h}^{\a}\leq 2\gep$ for all $k\in[2n_0-A,2n_0]$.
By induction one therefore gets that $Z_{N,h}^{\a}\leq 2  C_1$ for all $N$, such that
$\tf^{\a}(\gb,h)=0$, which gives a contradiction with the definition of $\hca$. 

\medskip
\textbullet\ We first start to show that for any $p\in[n_0+1, 2n_0]$, one has $Z_{p,h}^{\a}\leq 2 C_1$ .
We use the r.h.s. inequality of \eqref{quasirenew2} with $M=n_0$, and we divide the sum into two parts:
\begin{multline}
 Z_{p,h}^{\a}   \leq  C_1 \sum_{i=n_0-A}^{n_0} \sum_{j=n_0+1}^{p} Z_{i,h}^{\a} \K(j-i) Z_{p-j,h}^{\a} +
    C_1 \sum_{i=0}^{n_0-A-1} \sum_{j=n_0+1}^{p} Z_{i,h}^{\a} \K(j-i) Z_{p-j,h}^{\a} \\
   \leq   4\gep C_1^2 \sum_{n\geq1}n\K(n) + 4C_1^{3} \sum_{n\geq A} n\K(n) ,
\end{multline}
where we used the properties \eqref{cond1:propag}-\eqref{cond2:propag},
and the fact that $\K(j-i)$ appears at most $j-i$ times.
Thus we have $Z_{p,h}\leq 2C_1$ for $p\in[n_0+1, 2n_0]$ provided that
\begin{equation}
    \gep\leq (4C_1\bE[\tau_1])^{-1} \indent \text{ and } \indent \sum_{n\geq A} n\K(n) \leq (4C_1^{2})^{-1} ,
\label{eq:condAeps1}
\end{equation}
and we have the property \eqref{cond1:propag} with $n_0$ replaced by $2n_0$. 

\medskip

\textbullet\ We now show that $Z_{p,h}^{\a}\leq 2\gep$ for all $p\in[2n_0-A, 2n_0]$.
Again, we use the r.h.s. inequality of \eqref{quasirenew2} with $M=\lfloor p/2 \rfloor$, and
the properties \eqref{cond1:propag}-\eqref{cond2:propag} to get
\begin{multline}
 Z_{p,h}^{\a}   \leq  C_1 \sum_{i=\floor{p/2}-A/2 }^{\floor{p/2}} 
              \sum_{j=\floor{p/2}+1}^{\floor{p/2}+A/2} Z_{i,h}^{\a} \K(j-i) Z_{p-j,h}^{\a} +
                   C_1 \sumtwo{i<\floor{p/2}-A/2}{or \ j>\floor{p/2}+A/2} Z_{i,h}^{\a} \K(j-i) Z_{p-j,h}^{\a} \\
   \leq  4\gep^2 C_1 \sum_{n\geq1}n\K(n)  +  4 C_1^{2}\sum_{n\geq A/2} n\K(n),
\end{multline}
where we also used that we have $i,p-j \in [n_0-A,n_0]$ in the first sum
(since $p\in{[2n_0-A,} 2n_0]$), and $j-i\geq A/2$
in the second sum.
Thus we have $Z_{p,h}\leq 2\gep$ for $p\in[2n_0-A, 2n_0]$ provided that
\begin{equation}
 \gep\leq (4C_1\bE[\tau_1])^{-1}  \indent \text{ and } \indent \sum_{n\geq A/2} n\K(n) \leq (4C_1^2)^{-1} \gep ,
\label{eq:condAeps2}
\end{equation}
and we have the property \eqref{cond2:propag} with $n_0$ replaced by $2n_0$. 
\end{proof}

Claim \ref{claim:Znc}
controls directly the partition function,
instead of its Laplace transform as in Proposition \ref{prop:laplace}.
We emphasize that this improvement can be very useful, because it 
allows us to compare
$\Znc\Enc[\gd_i]$ with $\bP(i\in\tau)$,
analogously with \cite[Proposition 3.2]{BThier}.
For example an easy computation (expanding the exponential) gives that
\begin{equation}
\bE\left[ e^{ c_2 u \sum_{n=1}^{N} \gd_n } \ind_{\{N\in\tau\}} \right] \leq
 \Zna  =  \bE\left[ \exp\left( u \sum_{n=1}^{N} \gd_n \right) e^{H_{N,\hca}^{\a}}\right]
\leq \bE\left[ e^{ c_2 u \sum_{n=1}^{N} \gd_n } \ind_{\{N\in\tau\}} \right],
\end{equation} 
which gives more directly Theorem \ref{thm:expo_ann}.

\section{Proof of the results on the disordered system}
\label{sec:proofcorrel}

\subsection{The case of summable correlations, proof of Theorem \ref{thm:smooth}}
\label{sec:smoothcorrel}

As we saw in Section \ref{sec:anncorrel}, the annealed model is well-defined under the Assumption \ref{hyphyp}, only with $a>1$, when correlations are absolutely summable.

The proof of Theorem \ref{thm:smooth}, is very similar to what is done in \cite{GT05} for the case of independent variable. The main idea is to stand at $h_c(\gb)$ ($h_c(\gb)\geq \hca(\gb)>-\infty$ since the correlations are summable), and to get a lower bound for $\tf(\gb,h_c(\gb))$ involving $\tf(\gb,h)$ by choosing a suitable localization strategy for the polymer to adopt, and  computing the contribution to the free energy of this strategy.  This is inspired by what is done in \cite[Ch. 6]{GBbook} to bound the critical point of the random copolymer model.  More precisely one gives a definition of a "good block'', supposed to be favorable to localization in that the $\omega_i$ are sufficiently positive, and analyses the contribution of the strategy of aiming only at the good blocks. The main difficulty is here to get good estimates on the probability of having a "good" block

Let us fix some $l\in\N$ (to be optimized later), take $n\in\bbN$ and let 
$\cI\subset \{1,\dots,n\}$, which is supposed to denote
the set of indexes corresponding to "good blocks'' of size $l$, and we order its elements: $\cI=\{i_p\}_{p\in\bbN}$ with $i_1<i_2<\cdots$.
We then divide a system of size $nl$ into $n$ blocks of size $l$,
and denote $Z_{l,h}^{\go,(k)}$ the (pinned) partition function on the
$k^{\rm th}$ block of size $l$, that is $Z_{l,h}^{\go,(k)}= Z_{l,h}^{\theta^{(k-1)l}\go,\gb}$
($\theta$ being the shift operator, \textit{\textit{i.e.}\ } $\theta^{p}\go:= (\go_{n+p})_{n\geq 0}$).

For any fixed $\go$ and $n\in \bbN$, we denote $\cI_n=\cI\cap[0,n]$, so that targeting only the blocks in $\cI_n$ gives 
\begin{equation}
 Z_{nl,h}^{\go,\gb} \geq  \K((n-i_{|\cI_n|})l) \prod_{k=1}^{|\cI_n|} \K( (i_k-i_{k-1}-1) l)
   \prod_{k\in\cI_n} Z_{l,h}^{\go,\gb,(k)},
\end{equation}
with the convention that $\K(0):=1$. Then if $\gep>0$ is fixed (meant to be small), taking $l$ large enough so that $\log \K(kl) \geq -(1+\gep)(1+\ga)\log(kl)$ for all $k\geq 0$, one has
\begin{multline}
  \frac{1}{nl} \log Z_{nl,h}^{\go,\gb} \\
\geq  \frac{1}{nl} \sum_{k\in\cI_n} \log Z_{l,h}^{\go,(k)}
       - (1+\gep)(1+\ga)  \frac{1}{nl} \left(\log ((n-i_{|\cI_n|})l) + \sum_{k=1}^{\cI_n} \log ( (i_k-i_{k-1}-1) l) \right)\\
  \geq \frac{1}{n} \sum_{k\in\cI_n} \frac{1}{l}\log Z_{l,h}^{\go,(k)}
        - (1+\gep)(1+\ga) \frac{1}{l} \frac{|\cI_n|+1}{n}\log \left( \frac{n}{|\cI_n|+1} -1\right),
\label{eq:lowZhc}
\end{multline}
where we used Jensen inequality in the last inequality (which only means that the entropic cost of targeting the blocks of $\cI_n$ is maximal when all its elements are equally distant). Note that \eqref{eq:lowZhc} is very general, and it is useful to derive some results on the free energy, choosing the appropriate definition for an environment to be favorable
(and thus the blocks to be aimed), and the appropriate size of the blocks
(see Section \ref{sec:nonsummable} for another example of application).

We fix $\gb>0$, and set $u:=h-h_c(\gb)$. Then, fix $\gep>0$, and define the events
\begin{equation}
 \cA_l^{(k)} = \left\{ Z_{l,h_c(\gb)}^{\go,(k)} \geq \exp\left( (1-\gep)l\, \tf(\gb, h_c(\gb)+u) \right)\right\},
\end{equation}  
and
define $\cI_n$ the set of favorable blocks
\begin{equation}
\cI(\go):=\{  k\in\bbN\ :\ \cA_l^{(k)} \text{ is verified}\}.
\end{equation}

Then taking $l$ large enough so that \eqref{eq:lowZhc} is valid for the $\gep$ chosen above,
one has
\begin{equation}
\frac{1}{nl} \log Z_{nl,h}^{\go,\gb} \geq   \frac{|\cI_n|}{n} (1-\gep) \tf(\gb, h_c(\gb)+u)  -
    (1+\gep)(1+\ga) \frac{1}{l} \frac{|\cI_n|+1}{n}\log \left( \frac{n}{|\cI_n|+1} -1\right).
\end{equation}
We also note $p_l:=\bbP(\cA_l^{(1)})=\bbP(1\in\cI_n)$, so that one has that $\bbP$-a.s.
$\lim_{n\to\infty} \frac{1}{n}|\cI_n|=p_l$, thanks to
 Birkhoff's Ergodic Theorem
(cf.  \cite[Chap. 2]{Nadk}).
Then, letting $n$ go to infinity, one has
\begin{multline}
  0=\tf(\gb,h_c(\gb))\geq p_l (1-\gep)\tf(\gb, h_c(\gb)+u) -  (1+\gep)(1+\ga)p_l \frac{1}{l}\log ( p_l^{-1}-1)\\
   \geq p_l\left( (1-\gep)\tf(\gb, h_c(\gb)+u) +  (1+2\gep)(1+\ga)\frac{1}{l}\log ( p_l) \right),
\label{eq:up_smooth}
\end{multline} 
the second inequality coming from the fact that $p_l^{-1}$ is large for large $l$.
\smallskip

We now give a bound on $p_l$, with the same change of measure
technique used in the proof of Lemma \ref{lem:shiftgauss}.
We consider the measure $\bar\bbP$ on $\{\go_1,\ldots,\go_{l}\}$ which is absolutely continuous
with respect to $\bbP$, and consists in translating the $\go_i$'s of $u/\gb$, without changing the
correlation matrix $\gU$.
Then, using that
$l^{-1} \log Z_{l,h_c(\gb)}^{\go,\gb}$ converges to $\tf(\gb, h_c (\gb) + u)$ in $\bar\bbP$-probability as $l$ goes to infinity,
we have that $\bar \bbP (\cA_l^{(1)}) \geq 1-\gep$, for $l$ sufficiently large.
We recall the classic entropy inequality
\begin{equation}
\bbP(\cA)\geq \bar \bbP(\cA) \exp\left( -\frac{1}{\bar\bbP(\cA)} ({\rm H}(\bar\bbP|\bbP)+e^{-1}) \right),
\label{eqentropy}
\end{equation}
with $\rm{H}(\bar\bbP|\bbP)$ the relative entropy of $\bar \bbP$ w.r.t. $\bbP$. After some straightforward computation, one gets
${\rm H}(\bar\bbP|\bbP)= \frac{u^2}{2\gb^2}  \langle \gU^{-1}
\mathbf{1}_l, \mathbf{1}_l\rangle$, where $\mathbf{1}_l$ is the vector whose $l$ elements are all equal to $1$.

From Lemma \ref{lem:entropy2} (which needs $\gU$ to be invertible), one directly has that
${\rm H}(\bar\bbP|\bbP) 
   =  (1+o(1))\frac{u^2}{2\gU_{\infty}\gb^2}\, l $,
so that for $l$ large one gets that
\begin{equation}
\label{lowboundpl}
\frac{1}{l}\log p_l \geq -(1+\gep)\frac{1}{l} (1-\gep)^{-1}{\rm H}(\bar\bbP|\bbP)
 \geq -\frac{1+2\gep}{1-\gep}\frac{u^2}{2\gU_{\infty}\gb^2}.
\end{equation} 
This inequality, combined with \eqref{eq:up_smooth}, gives
\begin{equation}
 \tf(\gb, h_c(\gb)+u) \leq  -\frac{1+2\gep}{1-\gep}(1+\ga)\frac{1}{l}\log  p_l \leq
\left(\frac{1+2\gep}{1-\gep}\right)^2 \frac{1+\ga}{2\gU_{\infty}\gb^2} u^2,
\end{equation} 
which, thanks to the arbitrariness of $\gep$, concludes the proof.
\qed

\subsection{The case of non-summable correlations, proof of Theorem \ref{thm:xi<1}}
\label{sec:nonsummable}

This theorem is the non-hierarchical analogue of \cite[Theorem 3.8]{BThier}. But because there are some technical differences, we include the proof here for the sake of completeness.

\begin{proof}
  The idea is to lower bound the partition function by exhibiting a
  suitable localization strategy for the polymer, that consists in aiming
  at "good'' blocks, \textit{i.e.}\  blocks where $\go_i$ is very large.
  We then compute the contribution to the free energy of this strategy,
  in the spirit of \eqref{eq:lowZhc}.
  For $ a <1$ (non-summable correlations), it is a lot
  easier to find such large block (see Lemma \ref{lem:goodblock} to be
  compared with the independent case).  In this sense, the behavior of
  the system is qualitatively different from the $ a > 1$ case.

Clearly, it is sufficient to prove the claim for $h$ negative and large enough in absolute value. Let us fix $h$ negative with $|h|$ large and take $l=l(h)\in\bbN$, to be chosen later. Recall \eqref{eq:lowZhc}, and define
\begin{equation}
 \cA_l^{(k)}:=\left\{ \text{for all }  i\in [(k-1)l,kl]\cap\bbN, \text{ one has } \gb\go_i+h \geq |h|  \right\},
\end{equation}
and as in Section \ref{sec:smoothcorrel} the set of favorable blocks
$\cI_n$, and $p_l:=\bbP(\cA_l^{(1)})=\bbP(1\in\cI_n)$.

One notices that $ Z_{l,h}^{\go,(k)} \geq  Z_{l,|h|}^{\rm pur}$
for all $k\in\cI_n$, so that provided that $l$ is large enough, one has
$l^{-1}\log Z_{l,|h|}^{\p,\rm pure}\geq \frac12 \tf(|h|)$. Therefore, from \eqref{eq:lowZhc}, if $l$ is large enough so that the above inequality is valid, and letting $n$ goes to infinity, we get
$\bbP$-a.s.
\begin{equation}
 \tf(\gb,h)\geq \frac{p_l}{2} \tf (|h|) -C p_l \frac{1}{l}\log( p_l^{-1}-1)
   \geq p_l\left( c|h|+c' \frac{1}{l}\log p_l \right),
\label{eq:firstlowboundF}
\end{equation}
where we used that $\bbP$-a.s. $\lim_{n\to\infty} \frac{1}{n}|\cI_n|=p_l$, because of Birkhoff's Ergodic Theorem (cf.  \cite[Chap. 2]{Nadk}). The second inequality comes from the fact that, for $|h|\geq 1$, one has $\tf(|h|)\geq cst.\, |h|$, and that $p_l^{-1}$ is large if $l$ is large.

\smallskip
It then remains to estimate the probability $p_l$.
\begin{lemma}
\label{lemmamax}
Under Assumption \ref{hyphyp} with $a<1$, if correlations are non-negative, there exist two constants $c,C>0$ such that for every $l\in\N$ and $A\geq C(\log l)^{1/2}$ one has
\begin{equation}
 \bbP\left( \forall i\in\{1,\ldots,l\},\ \go_i\geq A \right)\geq c^{-1}\, \exp\left( -c A^2 l^{ a } \right).
\end{equation}
\label{lem:goodblock}
\end{lemma}

From this Lemma, that we prove in Appendix \ref{app} (Lemma \ref{lem:shiftgauss}), and choosing $l$ such that $\sqrt{\log l} \leq 2|h|/(C\gb)$, one gets that
\begin{equation}
  p_l= \bbP\big( \forall i\in\{1,\ldots,l\},\ \go_i\geq 2|h|/\gb \big)
    \geq c^{-1} \, \exp\left( -c l^{ a } h^2/\gb^2 \right).
\end{equation}
Then in view of \eqref{eq:firstlowboundF} one chooses $l= (\bar C |h|/\gb^2)^{1/(1- a )} $
(this is compatible with the condition $\sqrt{\log l}
\leq 2|h|/(C\gb)$ if $|h|$ is large enough) so that one gets
$c|h|+c'l^{-1}\log p_l \ge c|h|/2 \geq c/2$, provided that $\bar C$ is
large enough.  And \eqref{eq:firstlowboundF} finally gives with this
choice of $l$
\begin{equation}
 \tf(\gb,h)\geq cst.\, \exp\left( - c l^{ a } h^2/\gb^2  \right)
    \geq cst.\,  \exp\left( - c' |h| \left(|h|/\gb^2\right)^{1/(1- a )}  \right).
\end{equation} 
\end{proof}

{\bf Acknowledgement:} The author is very much indebted to Fabio Toninelli for his precious advice during the preparation of this paper, and would also like to thank the anonymous referee whose comments helped to improve the clarity and generality of the paper. This work was initiated during the author’s doctorate at the Physics Department of \'Ecole Normale  Sup\'erieure de Lyon, and its hospitality and support is gratefully acknowledged.

\begin{appendix}

\section{Estimates on correlated Gaussian sequences}
\label{app}

In this Appendix, we give some estimates on the probability for a long-range correlated Gaussian vector to
be componentwise larger
than some fixed value (see Lemma \ref{lem:shiftgauss}).
These estimates lies on the study of the relative
entropy of two translated correlated Gaussian vectors.
Let $\W=\{\W_n\}_{n\in\N}$ be a stationary Gaussian process, centered and with unitary variance, and with covariance matrix denoted by $\gU$. 
We write $(\rho_k)_{k\geq0}$ the correlation function, such that
$\gU_{ij}=\bbE[\W_i\W_j]=\rho_{|i-j|}$. Let $\gU_l$ denote the restricted correlation matrix, that is the correlation matrix of the Gaussian vector $\W^{(l)}:=(\W_1,\ldots,\W_l)$, which is symmetric positive definite.

We recall Assumption \ref{hyphyp}, that tells that the correlations are power-law decaying, {\sl i.e.}\ that $\rho_k\sim c_0 k^{- a }$ for some constants $c_0$ and $ a >0$.


\subsection{Entropic cost of shifting a Gaussian vector.}

In Section \ref{sec:smoothcorrel}, and in Lemma \ref{lem:shiftgauss}, one has to estimate the entropic cost of shifting the Gaussian correlated vector $\W^{(l)}$ by some vector $V$, $V$ being chosen to be $\ind_l$, the vector of size $l$ constituted of only $1$, or $U$, the Perron-Frobenius eigenvector of $\gU$ (if the entries of $\gU$ are non-negative). It appears after a short computation that the relative entropy of the two translated Gaussian vector of is $\frac12\langle\gU^{-1} V ,V \rangle$. We therefore give the two following Lemmas that estimate this quantity, one regarding the case $a>1$ (summable correlations), the other one the case $a<1$ (non-summable correlations).

\begin{lemma}[Summable correlations]
 \label{lem:entropy2}
Under Assumption \ref{hyphyp} with $a>1$, and if $\gU$ is invertible, then one has
\begin{equation}
\langle \gU_l^{-1} \ind_{l} ,\ind_l \rangle \stackrel{l\to\infty}{=}
    (1+o(1)) (\gU_\infty)^{-1} l,
\end{equation} 
where $\gU_{\infty}:= 1+2 \sum_{k\in\bbN} \rho_k$ (and $\ind_l$ was defined above).
\end{lemma}
Note that this Lemma is actually valid under the weaker assumption that $\sum |\rho_k|<+\infty$, $\gU$ still having to be invertible.

In the case $a<1$, we actually need the extra assumption that correlations are non-negative. 
We then note $\mu$ the maximal (Perron-Frobenius) eigenvalue of $\gU_l$, so that thanks to the Perron-Frobenius theorem we can take $U$ an eigenvector associated to this eigenvalue with $U_i>0$ for all
$i\in\{1,\ldots,l\}$. Up to a multiplication, we can choose $U$ such that $\min_{i\in\{1,\ldots,l\}}U_i=1$.

\begin{lemma}[Non-summable correlations]
 \label{lem:entropy1}
Under Assumption \ref{hyphyp} with $a<1$, and if $\rho_k\geq 0$ for all $k\geq 0$, one has that $\ind_l\leq U \leq c\ind_l$, where the inequality is componentwise.
Moreover, there exists a constant $c_4>0$ such that for all $l\in \bbN$ one has $ c_4^{-1} l^{1- a }\leq \mu\leq c_4 l^{1- a }$,
and therefore
\begin{equation}
c_4^{-1} l^{ a }\leq \langle \gU_l^{-1} U ,U \rangle \leq c c_4^{-1} l^{ a }.
\end{equation} 
\end{lemma}

Note that here, it is difficult to get directly an estimate on $\langle \gU_l^{-1} \ind_l,\ind_l \rangle$. The case $a=1$ is left aside, but one would get the same type of result, with $l^{a}$ replaced by $l/\log l$.

\begin{proof}[Proof of Lemma \ref{lem:entropy2}]
 The proof is classical, since we deal with Toeplitz matrices, and we include it here briefly, for the sake of completeness. The idea is to approximate $\gU_l$ by the appropriate circulant matrix $\Lambda_l$
\begin{equation}
 \Lambda_l:= \left(
\begin{array}{ccccccccc}
 \rho_0 & \cdots & \rho_m &        &        &        & \rho_m & \cdots & \rho_1 \\
 \vdots &        &        & \ddots &        &        &        & \ddots & \vdots \\
 \rho_m &        &        &        &        &        &    0   &        & \rho_m \\
        & \ddots &        &        &        & \ddots &        &        &         \\
        &        & \rho_m & \cdots & \rho_0 & \cdots & \rho_m &        &        \\
        &        &        & \ddots &        &        &        & \ddots &        \\
 \rho_m &        &   0    &        &        &        &        &        & \rho_m \\
 \vdots & \ddots &        &        &        & \ddots &        &        & \vdots \\
 \rho_1 & \cdots & \rho_m &        &        &        & \rho_m & \cdots & \rho_0
\end{array}
\right) ,
 \quad \text{ with } m= \lfloor \sqrt{l} \rfloor.
\end{equation}
One has that $\gU_l$ and $\Lambda_l$ are asymptotically equivalent, in the sense that their respective operator norms are bounded, uniformly in $l$ (thanks to the summability of the correlations), and that the Hilbert-Schmidt norm $||\cdot||_{\rm HS}$ of the difference $\gU_l-\Lambda_l$ verifies
\begin{equation}
 ||\gU_l-\Lambda_l||_{\rm HS}^2 :=  \frac{1}{l} \sum_{i,j}^l (\gU_{ij}-\Lambda_{ij})^2
    \leq \frac{c}{l} \left( \sum_{i=1}^l \sum_{k\geq m} \rho_k^2 + \sum_{i=1}^m \sum_{k=1}^{m}\rho_k^2 \right)
\stackrel{l\to\infty}{\to}0.
\label{equivmatrix}
\end{equation} 
For the convergence, we used that $m\ll l$, and the summability of the correlations.
One notices that $\ind_l$ is an eigenvector of $\Lambda_l$, and that
$\Lambda_l \ind_l = \upsilon_l \ind_l$, where $\upsilon_l:= 1+2\sum_{k=1}^m \rho_k$, which converges
to $\gU_{\infty}$. Then we use the idea that, as the operator norms of $\gU_l^{-1}$ and of $\Lambda_l^{-1}$ are asymptotically bounded, $\gU_l^{-1}$ and $\Lambda_l^{-1}$ are also asymptotically equivalent. One has 
\begin{equation}
 |\langle (\gU_l^{-1}-\Lambda_l^{-1})\ind_l,\ind_l\rangle|
     = \upsilon_l^{-1} |\langle \gU_l^{-1}(\gU_l-\Lambda_l)\ind_l,\ind_l\rangle|
  \leq l\,\upsilon_l^{-1}\,  |||\gU_l^{-1}|||\, ||\gU_l-\Lambda_l||_{\rm HS}.
\end{equation}
Therefore $\langle \gU_l^{-1} \ind_{l} ,\ind_l \rangle
     = \langle \Lambda_l^{-1} \ind_{l} ,\ind_l \rangle +o(l)=(1+o(1))\upsilon_l^{-1} l$,
which concludes the proof since $b_l\stackrel{l\to\infty}{\to}\gU_{\infty}$.
\end{proof}

\begin{proof}[Proof of Lemma \ref{lem:entropy1}]

We remark that the idea of the proof of Lemma \ref{lem:entropy2} would also work if $ a >1/2$ (and without the assumption of non-negativity), because in that case $\sum \rho_k^{2}<\infty$, and \eqref{equivmatrix} would still be valid. It is however difficult to adapt this proof to the $ a \leq 1/2$ case, and that is why we develop the following technique, that gives estimates on the eigenvector associated to the largest eigenvalue of $\gU_l^{-1}$.

Let us consider the Perron-Frobenius eigenvector $U$
of $\gU_l$, with eigenvalue $\mu$, as defined above:
we have that $U_i>0$ for all $i\in\{1,\ldots,l\}$, and we choose $U$ such that $\min_{i\in\{1,\ldots,l\}}U_i=1$.
Let us stress that one has, in a classical way
\begin{equation}
\label{boundlambda}
 \begin{split}
\mu \geq \min_{i\in\{1,\ldots,l\}} \sum_{j=1}^{l} \gU_{ij} & \geq
    c l^{1- a }, \\
\mu \leq \max_{i\in\{1,\ldots,l\}} \sum_{j=1}^{l} \gU_{ij} & \leq
    C l^{1- a },
 \end{split}
\end{equation} 
where we used the assumption \eqref{hyphyp} on the form of the correlations, and that $ a <1$. Then one has $\langle \gU_l^{-1} U ,U \rangle = \mu^{-1} \langle  U ,U \rangle $, so that we are left to show that the Perron-Frobenius eigenvector $U$ is actually close to the vector $\ind_l$. One actually shows that $\ind_l \leq U\leq c \ind_l$ where the inequality is componentwise, so that $cl\leq\langle  U ,U \rangle\leq c'l$, and it concludes the proof thanks to \eqref{boundlambda}.

\smallskip

We now prove that $U_{\infty}:=\max_{i\in\{1,\dots,n\}}U_i \leq c$ (we already have ${\min_{i\in\{1,\ldots,l\}}U_i=1}$).
Let us show that for $i<j$
\begin{equation}
 \label{intermU}
|U_i-U_j|\leq c \frac{|j-i|^{1- a }}{n^{1- a }} U_{\infty}.
\end{equation} 
One writes the relation $(\gU_lU)_a=\mu U_a$ for $a=i,j$, and gets
\begin{multline}
\label{UiUj}
\mu |U_i-U_j| =\left| \sum_{k=1}^{l} (\gU_{ik}-\gU_{jk})U_k\right|\\
 \leq U_{\infty}   \sum_{k=1}^{l} (\gU_{ik}-\gU_{jk})\ind_{\{\gU_{ik}>\gU_{jk}\}}
     + U_{\infty}   \sum_{k=1}^{l} (\gU_{ik}-\gU_{jk})\ind_{\{\gU_{ik}>\gU_{jk}\}}.
\end{multline}
From Assumption \ref{hyphyp} on the form of the correlations, there is some constant $C>0$ such that,
if $|j-i|\geq C$, then
one has $\rho_p > \rho_{p+|i-j|}$ for all  $p\geq |j-i|$.
Then one can write, in the case $i-j\geq C$, that
\begin{multline}
\sum_{k=0}^{l} (\gU_{ik}-\gU_{jk})\ind_{\{\gU_{ik}>\gU_{jk}\}}
  \leq \sum_{p=j-i}^{i}(\rho_{p}-\rho_{p+j-i}) + 2\sum_{p=0}^{j-1} K_{p}
      +\sum_{p=j-i}^{2(j-i)} (\rho_{p}-\rho_{p-(j-i)}) \\
  \leq 2\sum_{p=0}^{2(j-i)} \rho_{p} \leq c |j-i|^{1- a }.
\end{multline}
The second term in \eqref{UiUj} is dealt with the same way by symmetry,
so that one finally has $\mu |U_i-U_j|\leq c U_{\infty} {|j-i|^{1- a }}$
for $|i-j|\geq C$. Inequality \eqref{intermU} follows for every $i,j\in\N$ by adjusting the constant.

\smallskip
Suppose that $U_{\infty}\geq 4$.  The relation \eqref{intermU} gives that the components of the vector $U$ cannot vary too much.
One chooses $i_0$ such that $U_{i_0}=U_{\infty}$, and from \eqref{intermU} one gets that for all $j\in \N$
\begin{equation}
 U_{\infty}-U_j\leq c \frac{|j-i_0|^{1- a }}{n^{1- a }} U_{\infty}.
\end{equation} 
There is therefore some $\gd>0$, such that having $|j-i_0|\leq \gd l$ implies that $U_j\geq \frac12 U_{\infty} (\geq 2)$. Then, take $j_0$ with $U_{j_0}=1$ so that from writing $(KU)_{j_0} = \mu U_{j_0}$ one gets
\begin{equation}
 \mu = \sum_{k=1}^{l} \gU_{j_0k}U_k
 \geq \sumtwo{k=1}{|k-k_0|\leq \gd l/2}^{l} \gU_{j_0k} \frac{U_{\infty}}{2}
  \geq \frac{U_{\infty}}{2} \frac{\gd}{2} c l^{1- a },
\end{equation} 
where we used in the last inequality that, from
Assumption \ref{hyphyp}, there exists a constant $c>0$ such that for all $k\in\{1,\ldots, l \}$ one has $\gU_{j_0k}\geq c l^{- a }$, since $|j_0-k|\leq l$. One then concludes that $U_{\infty}\leq cst.$ thanks to \eqref{boundlambda}.
\end{proof}

\subsection{Probability for a Gaussian vector to be componentwise large}

We prove the following Lemma
\begin{lemma}
Under Assumption \ref{hyphyp} with $a<1$, and if $\rho_k\geq 0$ for all $k\geq 0$, there exist two constants $c,C>0$ such that for every $l\in\N$, one has
\begin{equation}
 \bbP\left( \forall i\in\{1,\ldots,l\},\ \W_i\geq A \right)\geq
      c^{-1}\, \exp\left( -c (A \vee C\sqrt{\log l})^2  l^{ a } \right).
\end{equation}
\label{lem:shiftgauss}
\end{lemma}
This Lemma, taking $A\geq C \sqrt{\log l}$, gives directly Lemma \ref{lemmamax}. Setting $A=0$, one would also have an interesting statement, that is that, when $ a <1$, the probability that the Gaussian vector is componentwise non-negative does not decay exponentially fast
in the size of the vector, but stretched-exponentially.

\begin{proof}
First of all, note $\cA:=\{\forall i\in\{1,\ldots,l\},\ \W_i\geq A\}$.
Set $\bar{\bbP}$ the law $\bbP$ on $\{\W_1,\ldots,\W_l\}$,
where the $\W_i$'s have been translated by $B\times U$, where $B :=2 (A \vee C\sqrt{\log l})$ (the constant $C$ is chosen later), and $U$ is the Perron-Fr\"obenius vector of $\gU_l$, introduced in Lemma \ref{lem:entropy2}.
Under $\bar \bbP$, $\{\W_i\}_{i\in\{1,\ldots,l\}}$ is a Gaussian vector
of covariance matrix $\gU_l$, and such that
$\bar \bbE \W_i = B U_i \geq B$ for all $1\leq i\leq l$.
Then one uses 
the classical entropic inequality
\begin{equation}
\bbP(\cA)\geq \bar \bbP(\cA) \exp\left( -\bar\bbP(\cA)^{-1} (\mathbf{H}(\bar\bbP|\bbP)+e^{-1}) \right),
\label{eq:entropy}
\end{equation}
where $\mathbf{H}(\bar\bbP|\bbP):=
\bbE\left[ \frac{\dd \bar\bbP}{\dd \bbP} \log \frac{\dd \bar \bbP}{\dd \bbP}\right] $ denotes the relative entropy of $\tilde \bbP$ with respect to $\bbP$.

\smallskip
Note that $\bar\bbP(\cA)\geq \bbP\left( \min\limits_{i=1,\ldots,l} \W_i \geq A-B\right)
     = \bbP\left(  \max\limits_{i=1,\ldots,l} \W_i \leq B-A \right)$,
and that $B-A\geq C\sqrt{\log l}$.
One uses Slepian's Lemma that tells that if $\{\hat{\W}_i\}_{i\in\{1,\ldots,l\}}$ is a vector of IID standard Gaussian variables (whose law is denoted $\hat{\bbP}$), then one has
\begin{equation}
 \bbE\left[ \max_{i=1,\ldots,l} \W_i \right] \leq \hat{\bbE}\left[ \max_{i=1,\ldots,l} \hat{\W}_i \right] \leq c\sqrt{\log l}, 
\end{equation}
where the second inequality is classical.  Thus one gets
\begin{equation}
 \bbP\left(  \max_{i=1,\ldots,l} \W_i \geq 2c\sqrt{\log l} \right) \leq
     \frac{1}{2c\sqrt{\log l}} \bbE\left[ \max_{i=1,\ldots,l} \W_i \right] \leq 1/2.
\end{equation} 
In the end, one chooses the constant $C$ such that
$ \bbP\left( \max_{i=1,\ldots,l} \W_i \leq 
  C\sqrt{\log l}  \right) \geq 1/2$ and one finally gets that
$\bar\bbP(\cA)\geq 1/2$.

\smallskip
One is then left with estimating the relative entropy
$\rm{H}(\bar\bbP|\bbP)$ in \eqref{eq:entropy}.  A straightforward
Gaussian computation gives
that $
\mathbf{H}(\bar\bbP|\bbP)= B^2 \langle \gU_l^{-1}
U, U\rangle.$
In the case $a<1$, Lemma \ref{lem:entropy2} gives  that
$\mathbf{H}(\bar\bbP|\bbP)\le  c B^2 l^{- a}$, which combined with \eqref{eq:entropy} gives the right bound.
\end{proof}

\end{appendix}

\bibliographystyle{plain}
\bibliography{bibliothese}

\end{document}